\documentclass[final,leqno]{siamltexmm}

\usepackage[T1]{fontenc}
\usepackage[utf8x]{inputenc}
\usepackage[english]{babel}
\usepackage[sfdefault,lf]{carlito}
\usepackage{amsmath}
\usepackage{graphicx}
\usepackage{booktabs}
\usepackage[colorinlistoftodos]{todonotes}

\usepackage{mathrsfs}
\usepackage{amssymb}
\usepackage{subcaption}
\usepackage{color}
\usepackage{algorithm}
\usepackage[noend]{algpseudocode}
\def\BState{\State\hskip-\ALG@thistlm}
\usepackage[rightcaption]{sidecap}
\usepackage{wrapfig} 
\newcommand{\smat}[4]{\begin{pmatrix} #1 &#2 \\ #3 &#4 \end{pmatrix}}
\newcommand{\svec}[2]{\begin{pmatrix} #1 \\ #2 \end{pmatrix}}

\newcommand{\norm}[1]{\left\lVert#1\right\rVert}
\newcommand{\off}[1]{}

\DeclareMathOperator{\Div}{div}
\DeclareMathOperator{\Grad}{grad}
\DeclareMathOperator*{\argsup}{argsup}

\newtheorem{conj}{Conjecture}[section]

\newcommand{\GG}[1]{{\color{blue} #1}}
\newcommand{\SB}[1]{{\color{purple} #1}}

\newtheorem{rem}{Remark}[section]
\title{Adaptive Anisotropic Total Variation - A Nonlinear Spectral Analysis}
\author{Shai Biton \and  Guy Gilboa} 
\begin{document}
\maketitle
\newcommand{\slugmaster}{\slugger{siims}{xxxx}{xx}{x}{x--x}}

\begin{abstract}
%% Written a new abstract. GG
A fundamental concept in solving inverse problems is the use of regularizers,  which yield more physical and less-oscillatory solutions. Total variation (TV) has been widely used as an edge-preserving regularizer. However, objects are often over-regularized by TV, becoming blob-like convex structures of low curvature. This phenomenon was explained mathematically in the analysis of Andreau et al. They have shown that a TV regularizer can spatially preserve perfectly sets which are nonlinear eigenfunctions of the form $\lambda u \in \partial J_{TV}(u)$, where $\partial J_{TV}(u)$ is the TV subdifferential. For TV, these shapes are convex sets of low-curvature.  
A compelling approach to better preserve structures is to use anisotropic functionals, which adapt the regularization in an image-driven manner, with strong regularization along edges and low across them. 
This follows earlier ideas of Weickert on anisotropic diffusion, which do not stem directly from functional minimization. Adaptive anisotropic TV (A$^2$TV) was successfully used in several studies in the past decade. However, until now there is no theory formulating the type of structures which can be perfectly preserved. In this study we address this question.

We rely on a recently developed theory of Burger et al on nonlinear spectral analysis of one-homogeneous functionals. We have that eigenfunction sets, admitting $\lambda u \in \partial J_{A^2TV}(u)$, are perfectly preserved under A$^2$TV-flow or minimization with $L^2$ square fidelity. We thus investigate these eigenfunctions theoretically and numerically. We prove non-convex sets can be eigenfunctions in certain conditions and provide numerical results which characterize well the relations between the degree of local anisotropy of the functional and the admitted maximal curvature. A nonlinear spectral representation is formulated, where shapes are well preserved and can be manipulated effectively. Finally, examples of possible applications related to shape manipulation and guided regularization of medical and depth data are shown.
  
\off{  
In image processing, over the years, we encountered many transforms such as fourier, cosine, wavelets and more. 
But while most of them either has no local knowledge in the frequency realm, and those which does, doesn't incorporate any data other than edges related; The TV-Transform and its spectrum gave the object and texture processing fields a great boost, and us the opportunity to research textures and shapes within the image in the image realm itself. Now, using the transform, we can separate textures or shapes to bands, and amplify or diminish them. But there is a downside to the TV Transform; While its ability to process textures is fantastic, its ability to process shapes is rather limited to convex and low curvature ones. In the real world, obviously, we have all kind of shapes and this is the problem this article is trying to solve. To do just that, we took an even more flexible functional, A$^2$TV-Adapted Anisotropic Total Variation, and plugged it to the transform only to create the A$^2$TV-Transform in which we believe to perfect the original transform to better process both texture and shapes. In this article we will show the theory behind the scenes and some cool applications of the transform.
}
\end{abstract} 

%\pagebreak
%\tableofcontents
%\pagebreak

\section{Introduction}
\label{sec_Intro}
%\subsection{TV in general}
%Writing about \\
%* variational methods\\
%* ROF model\\
%* TV Functional.\\
%* TV Flow.\\
%* previous work\\
%almost no equations. mostly words.

Regularization is a fundamental tool in solving ill-posed inverse problems. A classical edge-preserving regularizer is the total-variation functional. It has been used for solving various image processing and computer vision problems, such as denoising \cite{rof92}, segmentation \cite{iccv2015TVsegmentation}, optical flow \cite{zach2007TVopticalflow}, depth processing \cite{suwajanakorn2015depth} and more. See \cite{BurgerOsher2013_TV_Zoo,chambolle2010introduction} for detailed overviews addressing the properties of total variation and related applications.
% * <guy.gilboa@ee.technion.ac.il> 2018-06-05T15:16:09.428Z:
%
% ^.

A significant property of regularizers is the type of spatial shapes they can regularize well, with minimal distortions caused by the regularization process.
One such distortion is contrast change. It can be shown that any minimization of one-homogeneous functional, such as TV, with an $L^2$ square fidelity term (e.g. the ROF model \cite{rof92}) - reduces contrast for any calibrable set (or eigenfunction of the associated subdifferntial) \cite{Burger16}. Similarly this is true for gradient descent with respect to the regularizer, such as TV flow. Contrast reduction, however, can be avoided by using various methods such as inverse-scale-space \cite{iss}, nonlinear spectral 
filtering \cite{Gilboa_spectv_SIAM_2014}, employing other norms as fidelity terms \cite{Aubert_Aujol2008}, covariant least-squares \cite{deledalle2017clear} or adding a Bregman distance debiasing term \cite{brinkmann2017bias}.
A more significant byproduct of regularization is spatial smoothing and distortions of significant features of the signal.  In non-linear processing, these distortion effects are not simple to quantify and analyze. However, it is possible to characterize precisely which shapes maintain their geometrical structure, up to a contrast change, in regularization with $L^2$  fidelity or when gradient descent is applied. 

For the TV functional, this topic has been thoroughly investigated by Andreau and colleagues \cite{tvFlowAndrea2001}. It was shown that convex sets, which are smooth enough, termed \emph{calibrable sets}, preserve their shapes throughout gradient descent with respect to the TV functional (TV flow). A detailed presentation of these concepts and an outline of the main proofs are given below. Our work aims at extending the geometrical understanding to adaptive anisotropic TV. % \SB{(Maybe: adaptive anisotropic TV?)}. OK, GG
In \cite{Burger16} the issue of shape-preservation was generalized to any
absolutely one-homogeneous functional, where $J(\alpha u) = |\alpha|J(u)$, $\alpha \in \mathbb{R}$.  Denoting the subdifferential of $J$ at $u$ by $\partial J(u)$, 
it has been shown that nonlinear eigenfunctions, admitting 
$$ \lambda u \in \partial J(u), \qquad \lambda\in\mathbb{R},$$
preserve their structure through a gradient flow with respect to $J$ or in a minimization with $L^2$ square fidelity. Conditions for shape preservation of two structures, based on generalized inner-products were stated in \cite{gilboa2017semi} and \cite{schmidt2018inverse}.
In \cite{nossek2018flows} it was shown numerically that such eigenfunctions are very stable under noise and that one achieves state-of-the-art denoising performance in these cases.

For the TV case, the class of preserved shapes is quite restricted, including only convex rounded disk-like shapes, with low curvature. This implies that almost all geometrical shapes will be distorted during regularization. To allow better adaptivity of the regularizer to the signal at hand, several models were proposed to spatially change the degree and direction of the regularization at each point (for applications such as inpainting \cite{ex_Inpainting_example_duan2017introducing}, segmentation \cite{ex_Segmentation_example_hong2017multi}, filtering \cite{ex_filtering_example_krajsek2010diffusion}  and restoration \cite{ex_Denoising_example_min2017compressive}).
A canonical way to formalize this is by the following functional 
%\SB{(Maybe better to use a different letter then A, like D, because the immediate translation to our problem will show that $ourA=D^{1/2}$?)}
%\GG{Changed to mathcal\{A\}, maybe a little hard to see..} \SB{I can live with that ;)}
%\GG{I looked at \eqref{eq_lancoz} it looks exactly like this one.. PLEASE RESPOND (this means we can put regular $A$ and not $\mathcal{A}$.)}
%\SB{OK, The problem is that in here, as well as in  Eq. \eqref{eq_lancoz} the $A$ in lancoz is squared compare to our's. the problem with using another font is that you imply the this is another type of parameter where in reality its the same. therefore, I'm going back and suggesting to use both here and in Eq. \eqref{eq_lancoz} just the letter $B$ or maybe $\hat{A}$ or $\tilde{A}$ }
%\GG{I tried $\Sigma$, this should be changed also in \eqref{eq_lancoz}. When you introduce $A$ and the functional using it you should explain the relation between the matrices, (something like $A^{T}A=\Sigma$ ?)}
\begin{equation}
\label{eq_IntroductionAATVFunctional}
\mathcal{J} (u) = \int_\Omega \sqrt{\nabla u^T(x)R(x)\nabla u(x)} dx,
\end{equation}
where $R(x)\in\mathbb{R}^{2\times 2}$ is spatially adapted, such that the regularization is strong along the edges and weak across them. We term this functional as \emph{Adaptive Anisotripic Total Variation} (A$^2$TV), 
%\GG{Let me know if you're OK with A$^2$TV, then it should be replaced all along the paper..} \SB{Ok.. Done... I agree that it looks nitter ;)}
precise definitions are given in Section \ref{sec:AATV}. We add the term \emph{Adaptive} as \emph{Anisotropic TV} is often referred to evaluating the gradient magnitude in a rotationally varying manner (such as with $\ell^1$ metric) but with no spatial adaptive, image driven, changes. This functional was analyzed theoretically by Grasmair \& Lenzen  \cite{grasmair2010anisotropic}, where existence and uniqueness in several settings were shown. However, it is still an open problem -- which shapes such an energy can precisely preserve?  We attempt to provide first analytical and experimental answers in this work.

Decomposing a signal into its essential components is a long lasting problem in signal and image processing. In \cite{gilboa2013spectral,Gilboa_spectv_SIAM_2014}
a nonlinear spectral decomposition based on the TV functional was proposed. In \cite{Burger16} the approach was extended to absolutely one-homogeneous functionals. In specific settings, such as discrete one dimension, it was shown that one obtains a precise decomposition into nonlinear eigenfunctions. Moreover, the spectral components, which are a difference of two eigenfunctions, are orthogonal to each other. We will present here the 
A$^2$TV spectral framework and show its applicability.

\subsection{Outline and Main Contributions}
Our aim in this paper is to address the problem of shape preserving structures of the A$^2$TV functional. We extend the definitions of nonlinear eigenfunctions and calibrable sets to this model. In order to do it in a clear and self-containing manner, we provide in Section \ref{sec:tv_efs} a summary related to calibrable sets and eigenfunctions of TV. %This  includes also short proofs, which some readers may find them more accessible than the original proofs. 
As the topic is not trivial to non-experts, this section may be of some tutorial value. In Section \ref{sec:AATV} the A$^2$TV functional is presented and its eigenfunctions are analyzed. The associated spectral framework is presented in Section \ref{sec:spectral}. Ways to compute numerically the minimizations and flows are provided in Section \ref{sec:numerical}. Experiments illustrating the degree of shape preservation on toy examples are presented in Section \ref{sec:exp}. 
%Finally, we shown several applications using this functional for cartoon images, depth completion and regularization of acoustic data. \SB{ why not mentioning the CT-PET and the flowers?. could be: 
Finally, we give examples of possible uses based on the functional in a regularization context and through spectral filtering. These include depth inpainting, PET/CT Fusion and image manipulation.
%and regularization of \SB{is it ok to add opto-}acoustic data \GG{I think we should remove this sentence since we don't show any opto-acoustic data here.}. 

The main contributions of this work are: %\GG{Made some changes below}
\begin{enumerate}
\item We present A$^2$TV in a dual formulation setting, based on locally rotated and scaled gradient and divergence operators. Thereby providing a way to transfer some of the theory established for TV to the A$^2$TV functional.
\item The A$^2$TV is analyzed from a geometrical viewpoint. Two main properties of   eigenfunctions induced by the functional are investigated. We provide a proof for the convexity condition, showing that a set which is an eigenfunction %\GG{(removed "calibrable set")}
%(\SB{Problem with the use of the term 'calibrable sets'. calibrable sets has to be convex and round. 'A-Calibrable sets' can be non convex \& sharp} a calibrable set \SB{maybe:later to be defined as 'A-Calibrable sets'}) 
is not required to be convex, unlike the classical result for isotropic TV \cite{TVFlowInRN}. We show how a parameter which controls the anisotropy allows more complex nonconvex structure as the degree of anisotropy increases. Our condition turns out to coincide with a 
non-convexity measure for shape descriptors defined in \cite{peura1997efficiency}.
\item We suggest a conjecture related to the curvature bound, which stems from methodological numerical experiments.
%\item 
%\SB{added new one.. what do you think?}\GG{What is "generalized"? Not very clear.. when saying "for the first time" we need to be sure of that} \SB{I mean the dual form. and this is to the best of my knowledge and I guess yours too. can you ask other people? like Jean-Francois? or Cremer? If it is so; then I think its a major contribution.. if not then we can remove it.. ;) }
%\GG{We may say (without first time) 
%"}\SB{Agreed}
%We present the A$^2$TV for the first time in its generalised version, similar to the TV energy. Thereby providing a way to transfer some of the theory investigated on the TV to the A$^2$TV functional.
\item We present for the first time the attributes of spectral A$^2$TV. 
%a new perspective to the scale space spectrum by 
It is shown how nonlinear spectral filtering decomposes an image into textures and  objects parts in a new way, allowing objects to retain detailed geometry. We show the advantages of the A$^2$TV transform over the TV transform with respect to preserving well shapes of complex non-convex structures.
\item We suggest some applications which use the functional. The analysis and understanding of eigenfunctions of A$^2$TV allows new insights on the regularization functional and better employment of it. 
\end{enumerate}

% \section{Notation And Formulation}
% \GG{Not sure we need this section, we can put the essential things in section 3. Define only things we crucially depend on later.}
% \SB{Agreed}
% * bounded variation BV(Q)\\
% * perimeter\\
% * |A| - area in R2\\
% * characteristic functions\\

%\pagebreak
\subsection{Some Notation}
Let $Q$ be an open subset in $\mathbb{R}^N$. A function $u\in L^1(Q)$ with gradient $Du$ (in the sense of distribution) is a Radon measure with finite total variation in $Q$ is called a function of bounded variation and will be denoted by $u\in BV(Q)$.
A characteristic function of a set $E$ will be denoted by $\chi_E$.
For a finite perimeter $E$, one can define the essential boundary $\partial^*E$ as a countably $(N-1)$ rectifiable with finite $\mathscr{H}^{N-1}$.
We define $\theta(x,y)\in L^\infty_{|y|}(\mathbb{R}^N)$ as the density of $(x,y)$ with respect of $|y|$,
$$(x,y)(B) = \int_B \theta(x,y)d|y|,$$ 
for any Borel set $B\subseteq\mathbb{R}^N$.
In particular, if $E$ is bounded and has finite perimeter in $\mathbb{R}^N$, it follows that,
\begin{equation*}
J_{TV}(\chi_E) = \langle\Div\xi(x),\chi_E(x)\rangle = \langle\xi(x),-D\chi_E(x)\rangle = \int_{\partial^*E}\theta(\xi(x),-D\chi_E(x))d\mathscr{H}^{N-1}.
\end{equation*}

\subsection{Convexity measure}
In shape analysis there has been comprehensive research related to characterizing shapes by a variety of measures and geometric descriptors.
In \cite{peura1997efficiency} Peura and Iivarinen proposed a group of fundamental shape descriptors. Specifically, they proposed the following \emph{convexity measure} of a set $C$:
\begin{equation}
\label{eq:conv_C}
\mathcal{H}(C) := \frac{P(Conv(C))}{P(C)},
\end{equation}
where $P(C)$ is the perimeter of $C$ and $Conv(C)$ is the convex-hull of the set $C$. This measure has certain desired property as a shape descriptor.
It is in the range $(0,1]$ where it is $1$ if and only if the set $C$ is convex. It is also translation and rotation invariant. We will later see how this measure is directly related to characterizing eigenfunctions of the A$^2$TV functional. 

\newpage
\section{The Theory of TV Eigenfunctions}
\label{sec:tv_efs}
%As it was well explained in the introduction, the TV functional is a very interesting one as it able to give 
\subsection{Functional Definition}
The TV functional is defined by,
\begin{equation}
\label{eq_TVStrongDef}
J_{TV}(u) = \sup_{\xi \in C_c^\infty, \|\xi\|_\infty\leq 1}\int_\Omega u(x)\Div\xi(x)dx,
\end{equation}
on a domain $\Omega$, for $u\in BV(\Omega)$. For smooth functions $u$, this reduces to,
\begin{equation}
\label{eq_TVweakDef}
J_{TV}(u) = \int_\Omega |\nabla u(x)|dx.
\end{equation}

\subsection{Regularization based-on TV} 
In this paper we focus on gradient descent flows with respect to the regularizing functional. In the case of TV, this is referred to as \emph{TV flow}, creating an inherent nonlinear, edge-preserving, scale-space. 
It is closely related to the ROF model \cite{rof92}, where the regularization is coupled with $L^2$ square fidelity. For completeness, both are defined below.
%However, most of the results and analysis are valid for the regularization with $L^2$ square norm as well (ROF model). 

\subsubsection{The TV Flow} 
The TV flow regularizes an input signal $f(x)$ by performing a gradient descent with respect to TV. 
%Here, for $u(t;x)$, we interested in the equation,
We denote by $u(t;x)$ the solution of the following flow (with a time parameter $t$), 
\begin{equation}
\label{eq_TVflow}
\begin{array}{lcl}
% \frac{\partial u}{\partial t} = \Div \left( {\frac{\nabla{u}}{\|\nabla{u}\|_2}} \right) &on &(0,\infty)\times\Omega \\\\
\frac{\partial u}{\partial t} = -\Div \left( \xi(x) \right) &on &(0,\infty)\times\Omega, \\\\
\frac{\partial u}{\partial n} = 0  &in &(0,\infty)\times\Omega, \\\\
u(0;x) = f(x) &in &x\in\Omega,
\end{array}
\end{equation}
where $\xi$ admits $ \argsup_{\xi \in C_c^\infty, \|\xi\|_\infty\leq 1}\int_\Omega u(x)\Div\xi(x)dx$.
Hence $ \Div \left( \xi(x) \right) \in \partial J_{TV}(u)$ is a subgradient element and the flow is a gradient descent with respect to TV. Based on the ``chain rule'' for the differentiation of functionals (see Lemma 3.3 of Brezis \cite{brezis1973ope}), denoting $p= \Div \left( \xi(x) \right)$, we obtain the monotonic decrease of $J_{TV}(u(t))$ over time by
$$ \frac{d}{dt}J(u(t))=\langle p, \frac{\partial u}{\partial t}\rangle = -\| p\|^2 \le 0, \qquad a.e.$$
\subsubsection{ROF Model} 
The ROF model \cite{rof92} is the classical way of using TV to regularize an image $f$. The resulting image $u$ is a minimizer of the following energy,  
\begin{equation}
\label{eq_TVROFModel}
\mathscr{E}(u)= t \cdot J_{TV}(u) + \frac{1}{2}\|u-f\|_{L_2(\Omega)}^2.
\end{equation}
There is a strong resemblance between the TV flow and the ROF model. In the discrete 1D case it was shown that identical solutions are obtained (see e.g. \cite{steidl2004equivalence,Burger16}).
In this section we examine shapes which preserve their spatial structure (up to a contrast change) in a gradient flow or regularization. This type of functions admit a nonlinear eigenvalue problem, with respect to the subgradient of the functional. We refer to them as eigenfunctions.
\subsection{TV Eigenfunctions}
The theory of total variation regularization in $\mathbb{R}^2$   characterizes the functions for which the TV regularizer is well adapted to. Those are \emph{calibrable sets}, which are a specific family of \emph{eigenfunction}. Let us first define a nonlinear eigenfunction induced by a functional.

\begin{definition}[Eigenfunction induced by a functional $J$]
\label{def:ef}
A function $f$ is an eigenfunction induced by a functional $J$ if it admits the following nonlinear eigenvalue problem,
\begin{equation}
\label{eq_TVef}
\lambda f \in \partial J(f),
\end{equation}
%We refer to functions admitting \eqref{eq_TVef} as eigenfunctions, 
where $\lambda$ is the corresponding eigenvalue.
\end{definition}
%
%For the case of TV, $J=J_{TV}$ as defined in \eqref{eq_TVStrongDef}, we are mostly interested in special eigenfunctions, indicator functions of calibrable sets, which are defined hereafter in \ref{def:CalibrableSets}.
%% Already said above
%
%
%There are two main conditions that define the characteristics of eigenfunctions: The \emph{convexity} and the \emph{curvature} conditions.
%The paper of Bellettini et al \cite{TVFlowInRN} gives, in great detail, the theory behind the TV flow and its eigenfunction sets. 
%
% stuff\\
% \lambda = P(C)/\|C\|
%
% \begin{equation*}
% \max_{\forall x\in \partial^*C}\kappa(x)\leq\lambda_C
% \end{equation}
% Hence, $\forall D\supset\Omega$, we have $Per(\Omega) = \lambda_\Omega |\Omega| = \lambda_\Omega |\Omega\cap D|$ which gives us $Per(\Omega) \leq Per(D)$. Meaning  $\Omega$ must be convex.
% (if taking D as a superset of $\Omega$ means that it has a bigger perimeter, then $\Omega$ must be convex; otherwise $\Omega$ could enlarge its perimeter by, for example, choosing C shape instead of an O shape).\\
% 
%For the case of TV 
A deep theoretical research  related to the above eigenvalue problem has been conducted in \cite{TVFlowInRN} and \cite{bellettini2002total} for the TV functional. 
Here we outline the essential results and provide proofs which are crucial for better understanding the advantages of the A$^2$TV functional, as explained in Section \ref{sec:AATV}.
It was shown that eigenfunctions preserve their shape and only lose contrast under the TV flow. In fact, the loss of contrast is linear, where the TV flow \eqref{eq_TVflow} has the following analytic solution,
\begin{equation}
\label{eq_TVflowSolution}
u(t) = (1-t\lambda)^+f,
\end{equation}
for initial conditions $f$ admitting \eqref{eq_TVef}, $J=J_{TV}$. 
Moreover, these studies gave a precise geometrical characterization when $f$ is an indicator of a set $C\subset\mathbb{R}^2$, which is a bounded set of finite perimeter. The theory requires that the set $C$ be convex. A second condition (we refer to as a curvature condition) consists of a bound on the maximal curvature on the boundary of the set,
\begin{equation}
\label{eq_curvature}
\max_{\forall x\in \partial^*C}\kappa(x)\leq \frac{P(C)}{|C|},
\end{equation}
for a characteristic set $C$, where $P(C)$ is the perimeter of $C$ and $|C|$ is its area. In the section below we summarize the essential arguments and analysis for reaching the convexity condition and Eq. \eqref{eq_curvature}. This background will serve us in better understanding our new analysis for the adapted-anisotropic TV case. 

%convexity & curvature proofs
\subsubsection{The convexity condition}
In order to prove the convexity condition, let us first introduce several definitions. Let the scalar $\lambda_C$ be defined by,
\begin{equation}
\label{eq_lambdaTV}
\lambda_C = \frac{P(C)}{|C|}.
\end{equation}
%where, $C\subset\mathbb{R}^2$ is a bounded set of finite perimeter. 
We can now define a calibrable set.
\begin{definition}[Calibrable Set]
\label{def:CalibrableSets}
%Let $C\subseteq\mathbb{R}^2$ be a set of finite perimeter.
We say that $C$ is a $-$calibrable set if there exists a vector field $\xi^-_C:\mathbb{R}^2\to\mathbb{R}^2$ with the following properties:
\begin{enumerate}
\item $\xi^-_C\in L^2_{loc}(\mathbb{R}^2;\mathbb{R}^2)$ and  $\Div(\xi^-_C)\in L^2_{loc}(\mathbb{R}^2)$.
\item $|\xi^-_C|\leq 1$ almost everywhere in $C$.
\item $\Div(\xi^-_C)$ is constant on $C$.
\item $\theta(\xi^-_C,-D\chi_C)(x)=-1$ for $\mathscr{H}^1$-almost every $x\in\partial^*C$.
\end{enumerate}
\end{definition}
We say that $C$ is $+$calibrable if there exists a vector field $\xi^+_C:\mathbb{R}^2\to\mathbb{R}^2$ satisfying properties (1),(2),(3) and such that $\theta(\xi^+_C,-D\chi_C)(x)=1$.

Basically, condition (4) means that $\xi^-_C$ ($\xi^+_C$) is identical to the outer (inner) normalized gradient of $\chi_C$ for almost every $x\in\partial^*C$.\\

\begin{rem}
Let $C\subset\mathbb{R}^2$ be a bounded set of finite perimeter. Assume that $C$ is $-$calibrable and that $\mathbb{R}^2\diagdown C$ is $+$calibrable. Define
\begin{equation}
\label{xi_def}
\xi_C := -\Bigg\{
\begin{array}{lcl}
\xi^-_C                        & on &C\\
\xi^+_{\mathbb{R}^2\diagdown C}  & on &\mathbb{R}^2\diagdown C,
\end{array}
\end{equation}
 
then, $\xi_C \in L^\infty(\mathbb{R}^2;\mathbb{R}^2)$ and $\Div(\xi_C )\in L^\infty(\mathbb{R}^2)$.\\
\end{rem}

In \cite{TVFlowInRN} it was shown that for calibrable sets % $\Div(\xi^-_\Omega)$ is a constant $\lambda_\Omega$; on $\mathbb{R}^2\diagdown\Omega$, the $\Div(\xi^+_{\mathbb{R}^2\diagdown\Omega})$ is a constant zero.\\
%meaning,
 \begin{equation}
 \label{eq_divxi}
\Div(\xi_C) = \lambda_C\chi_C = \Bigg\{
\begin{array}{lcl}
\lambda_C    & on &C\\
0  & on &\mathbb{R}^2\diagdown C ,
\end{array}
\end{equation}
% \textbf{The convexity condition}
where $\lambda_C$ is defined by Eq. \eqref{eq_lambdaTV}.
We can now state the following convexity condition.
\begin{theorem}[The Convexity Condition]
\label{Theorem_Convex}
Let $C\subset\mathbb{R}^2$ be a bounded set of finite perimeter which is $-$calibrable such that $\mathbb{R}^2\diagdown C$ is $+$calibrable. Then,
\begin{enumerate}
\item $C$ is convex. 
\item $\forall D\subseteq C$, we have $\lambda_D \geq\lambda_C$.
\end{enumerate}
\end{theorem}

\begin{proof}
Let $\xi_C \in L^\infty(\mathbb{R}^2;\mathbb{R}^2), \| \xi_C \|\leq 1$ be the vector field defined by \eqref{xi_def}.
% $\xi:= \xi^+_{\mathbb{R}^2\diagdown\Omega}\chi_{\mathbb{R}^2\diagdown\Omega} + \xi^-_\Omega\chi_\Omega $.\\
Let $D$ be a set of a finite perimeter. Using the fact that $\lambda_C\chi_C =\Div(\xi_C )$, we have,
\[
\begin{aligned}
%Per(\Omega) = \lambda_\Omega |\Omega|
\lambda_C |C\cap D| 
&= \lambda_C\int_{\mathbb{R}^2}\chi_C\chi_D 
= \int_{\mathbb{R}^2}\lambda_C\chi_C\chi_D 
= \int_{\mathbb{R}^2}\Div(\xi_C)\chi_D \\
&= \int_D \Div(\xi_C)
=\underbrace{ \oint_D \xi_C\cdot \hat{\nu}_D}_{\substack{
\textrm{Divergence Theorem}\\ \hat{\nu}_D \textrm{ outer normal of $D$}}}
\leq \underbrace{\oint_D dl}_{\|\xi_C\|\leq 1,\hspace{1mm}\|\hat{\nu}_D\|=1} = Per(D)
\end{aligned}.
\]
Hence, $\forall D\supseteq C$, we have $Per(C) = \lambda_C |C| = \lambda_C |C\cap D|$ and therefore $Per(C) \leq Per(D)$. Therefore  $C$ must be convex
(If $C$ is not convex then we can choose $D$ as its convex hull  which has a smaller perimeter).
%(if $D$ as a superset of $C$ means that it has a bigger perimeter, then $C$ must be convex; otherwise $C$ could enlarge its perimeter by, for example, choosing C shape instead of an O shape).\\
While $\forall D\subseteq C$, we have $\lambda_C |D| = \lambda_C |C\cap D|$ yielding $\lambda_C\leq\lambda_D$.
\end{proof}

\subsubsection{The curvature condition}

For the curvature condition we shall define the following, 
\begin{definition}[$\mathscr{G}_\lambda$]
Given $\lambda\in\mathbb{R}$ we define the functional $\mathscr{G}_\lambda$ as \\
\begin{equation}
\mathscr{G}_\lambda(D) := P(D) -\lambda|D|,\hspace{1cm}  D\subseteq \mathbb{R}^2,\hspace{1cm} D \hspace{1mm}\text{of finite perimeter}.
\end{equation}
\end{definition}

The functional $\mathscr{G}_\lambda$ has the following properties.
\begin{proposition}
\label{Prop_G}
Let $C,D\subset\mathbb{R}^2$ be bounded sets of finite perimeter which are $-$calibrable and such that $\mathbb{R}^2\diagdown C$ and $\mathbb{R}^2\diagdown D$ is $+$calibrable. Also 
$\lambda_C\chi_C =\Div(\xi_C ), \hspace{1mm}\lambda_C = P(C)/|C|$ on $C$ and $\lambda_D\chi_D =\Div(\xi_D ), \hspace{1mm} \lambda_D = P(D)/|D|$ on $D$.
Then,
\begin{enumerate}
\item $\forall D\subseteq C$, the set $C$ minimizes $\mathscr{G}_{\lambda_C}(D) $. or,   $\mathscr{G}_{\lambda_C}(D)\geq\mathscr{G}_{\lambda_C}(C) $.
\item $\forall D\subseteq C$,  $\mathscr{G}_{\lambda_C}(D)\geq 0$.
\item $\mathscr{G}_{\lambda_C}(C)=0$.
\end{enumerate}
\end{proposition}

\begin{proof}
Since $D\subseteq C$, by Theorem \ref{Theorem_Convex}, it holds that $\lambda_D\geq\lambda_C$, therefore,
\[
\begin{aligned}
\mathscr{G}_{\lambda_C}(D) &=  P(D) -\lambda_C|D|
= \int_D (\Div(\xi_D)-\lambda_C)
= \int_D (\lambda_D-\lambda_C)
\geq \int_D (\lambda_C-\lambda_C) = 0  \\
&= \int_C (\Div(\xi_C)-\lambda_C)
= P(C) -\lambda_C|C|
=\mathscr{G}_{\lambda_C}(C).
\end{aligned}
\]
\end{proof}

% \textbf{The curvature condition}
Finally, the following curvature condition can be stated.
\begin{theorem}[The Curvature Condition]
\label{Theorem_Curv}
Let $C\subset\mathbb{R}^2$ be a bounded set of finite perimeter which is $-$calibrable and such that $\mathbb{R}^2\diagdown C$ is $+$calibrable. Then,
the following holds,
\begin{equation}
\label{eq:max_curvature_tv}
\max_{\forall x\in \partial^*C}\kappa(x)\leq\lambda_C,
\end{equation}
where $\kappa(x)$ is the curvature at point $x$ on the boundary of $C$.
\end{theorem}

\begin{rem}
This important property was proved in \cite{TVFlowInRN} using Prop. \ref{Prop_G}; but also in \cite{CurvProof_alter2005characterization,CurvProof_giusti1978equation,CurvProof_kawohl2006characterization} in aspects such as cheeger sets and anisotropic norms.
\end{rem}
We now turn to the main topic of the paper, concerned with the anisotropic functional, and see where insights gained in the TV theory can be extended and applied.

\section{The Adapted Anisotropic Total Variation (A$^2$TV) Functional}
\label{sec:AATV}
Several variants of anisotropic TV were proposed in different studies  for  applications such as 
inpaiting \cite{ex_Inpainting_example_duan2017introducing}, segmentation  \cite{ex_Segmentation_example_hong2017multi}, filtering \cite{ex_filtering_example_krajsek2010diffusion} and restoration \cite{ex_Denoising_example_min2017compressive} . The  formulation and analysis most relevant for our study is the one of Grasmair \& Lenzen \cite{grasmair2010anisotropic}. In \cite{grasmair2010anisotropic} the regularizer has been presented in the form of Eq. \eqref{eq_IntroductionAATVFunctional}. In order to define spatially adaptive operators, which are useful in the analysis, we slightly change the notations. We introduce a matrix $A(x)$ which is related to $R(x)$ by $R(x)=A^T(x) A(x)$. The functional in the strong-sense is thus defined by,
\begin{equation}
\label{eq_lancoz}
\mathcal{J} (u) = \int_\Omega \sqrt{\nabla u^T(x)A^T(x) A(x)\nabla u(x)} dx,
\end{equation}
where $u$ is a smooth function and $A(x)\in\mathbb{R}^{2\times 2}$ is a symmetric matrix, spatially adapted, $A\succ 0$. We have $A\in L^\infty$ and specifically in our applications $0<det(A)(x)\le1$.
We wish to represent this regularizer with a slightly different formulation, better suited for our framework. Therefore, we introduce variants of the gradient and divergence operators.
\subsection{Introducing $\Grad_A(\cdot)$ and $\Div_A(\cdot)$}
One way to view A$^2$TV, is as a total variation regularization applied on a deformed space.
Another way, is to apply regularization with space-dependent weights. 
Either way, it is beneficial to define modified variants of the gradient and divergence operators. This yields, in addition, a straightforward way to implement the A$^2$TV.  

We define the following gradient,
\begin{equation}
\label{eq_GradA}
\Grad_A = \nabla_A \triangleq A\nabla.
\end{equation}
One can obtain the appropriate divergence operator by finding the Hermitian adjoint of our adapted gradient, as the operators are connected by the identity,
\[ \langle \Grad_A u,\vec{v}\rangle  =  \langle u,- \Div_A \vec{v}\rangle. \]
To solve this, we first note that while our adapted gradient operates on a set of functions admitting the Neumann boundary condition, i.e.  , $\forall x \in \partial \Omega: \frac{\partial u}{\partial n}(x) =0 $, the adjoint operator  assumes vectors admitting the Dirichlet condition, i.e. $\forall x \in \partial \Omega: \vec{v}(x) = \begin{pmatrix} v_1(x) \\ v_2(x) \end{pmatrix} = \vec{0} $ . The adapted divergence is derived by,
\begin{align}
\langle\nabla_Au,\vec{v}\rangle 
 &= \int_\Omega A\nabla u \cdot\vec{v}dx = \int_\Omega \nabla u \cdot A^{T}\vec{v}dx
 = \underbrace{\int_{\partial\Omega} u \left(A^{T}\vec{v}\cdot\hat{n}\right) dx}_{=0} - \int_\Omega u \nabla\cdot\left(A^{T}\vec{v}\right) dx \notag\\
 &= \langle u,-\nabla_A^T\vec{v}\rangle \notag.
\end{align}
Since $A$ is symmetric, we obtain a definition for the adapted divergence,
\begin{equation}
\label{eq_DivA}
 \Div_A = \nabla_A^T \triangleq \nabla^T A^T = \nabla^T A.
\end{equation}

\subsection{Functional Definition}
The A$^2$TV functional is spatially adapted in order to preserve the structure of the objects within the image $u$. Using the above operators, the Adapted Anisotropic Total Variation functional in a domain $\Omega$ is given by,
\begin{equation}
\label{eq_J_AATV}
J_{A^2TV}(u)  = \int_\Omega| A(x)\nabla u | dx = \int_\Omega |\nabla_A u(x)|dx,
\end{equation}
for smooth functions $u$.
The weak-sense formulation for $u\in BV(\Omega)$  is,
\begin{equation}
\label{eq_AATVstrongDef}
J_{A^2TV}(u) = \sup_{\xi^A \in C_c^\infty, \|\xi^A\|_\infty\leq 1}\int_\Omega u(x)\Div_A\xi^A(x)dx,
\end{equation}
where $A(x)\in\mathbb{R}^{2\times 2}$  is the same one as in Eq. \eqref{eq_lancoz}.
There are many ways to choose the matrix  $A$. We focus on two of them below.
\subsection{Definition of the matrix $A$}
\label{DefA}
There are several options for how to compute the matrix $A$ in an image-driven manner, depending on the task at hand. 
We focus on the classical coherence-enhancing method presented by Weickert for anisotropic diffusion \cite{AnisoWeickert}. There is still an active research related to this interesting process with very useful applications recently suggested for shape inpainting and compression \cite{schmaltz2014understanding,peter2017turning}.
We then simplify it for sets in order to better analyze it theoretically and numerically. 

We are interested in a tensor $A$ that properly preserves edges of complex objects. We would like a strong regularization along edges and weaker one across them. In order to estimate locally the directions a smoothed structure tensor is used. The smoothed structure tensor is defined by the outer product,
 
\[ J_{\rho} (\nabla u_\sigma) = k_\rho * \left( \nabla u_\sigma \otimes \nabla u_\sigma\right) ,\]
where $\nabla u_\sigma$ is the gradient of an image $u$  smoothed by a Gaussian kernel $k_\sigma$  with variance $\sigma^2$ . The structure tensor itself is smoothed by a Gaussian kernel with variance $\rho^2$. In two dimensions, we can write a simple expression for the tensor as a $2 \times 2$ matrix for each point $x=(x_1,x_2)\in \Omega$,

\[ J_{\rho} (x) =
\begin{pmatrix}
(k_\rho* u_{x_1;\sigma}^2)(x) &(k_\rho* u_{x_1;\sigma}u_{x_2;\sigma})(x) \\
(k_\rho* u_{x_2;\sigma}u_{x_1;\sigma})(x) &(k_\rho* u_{x_2;\sigma}^2)(x) 
\end{pmatrix}.
 \]
This matrix has eigenvectors corresponding to the average direction of the gradient and the tangent at a neighborhood $\sigma$ around point $x$ and eigenvalues corresponding to the magnitude of each direction. In order to preserve structure, we should attenuate the eigenvalue corresponding to the gradient direction (or increase the eigenvalue corresponding to the tangent direction).
Let us now give a precise scheme for the tensor. We begin by looking at the eigen-decomposition of the smoothed structure tensor,
\[ J_\rho = V \Lambda V^{-1}, \]
where $V$ is a matrix whose columns, $v_1,v_2\in\mathbb{R}^2$, are the eigenvectors of $J_\rho$ , and   $\Lambda$ is a diagonal matrix with the corresponding eigenvalues,
%\GG{We use $D$ in two different contexts, as a set (like the set $C$ in section 2 and later) and here. We need to change one of them to a different letter (or mathcal $\mathcal{D}$, bold $\bf D$ ?). Added index for the different definitions, $D_1$, $D_2$}\SB{Im suggesting $\Lambda$.} %OK
%
\begin{equation}
\label{def_V_D}
\Lambda = \begin{pmatrix}
\mu_1 &0 \\ 0 &\mu_2
\end{pmatrix}; 
\hspace{1cm}
V = \left(v_1;v_2\right).
\end{equation}
 Let us denote the eigenvalues by $\mu_1$,$\mu_2$ , and assume $\mu_1 \ge \mu_2$ (where $\mu_i$ is the corresponding eigenvalue of $v_i$). We keep the original directions of the gradient and tangent, by keeping the matrix $V$ , and applying a modification on the matrix $\Lambda$ . We consider two modifications.  
 The first, a well-known modification suggested by Weickert,
 
 \begin{equation}
 \label{Eq_Wickert_D_tilde}
\Lambda_1 = \begin{pmatrix}
c\left(\frac{\mu_1}{\mu_{1,mean}},k\right) &0 \\ 0 &1
\end{pmatrix};
\hspace{1cm}
c(s;K) = \Bigg\{
\begin{array}{lc}
1  & s\leq 0\\
1-e^{-\frac{C_m}{(s/K)^m}}  &s>0
\end{array}.
\end{equation}
The matrix $A$ is then defined by, 
\begin{equation}
 \label{Eq_Wickert_A}
A=V\Lambda_1V^{-1}.
\end{equation}

This modification is highly instrumental for natural complex images, which include texture and noise. However, it has several parameters and it is not trivial to analyze. We would like to examine an image $f$ which is an indicator function of some set $C$. One can simplify the construction of $A$ considerably, by using a single parameter $a$, controlling the degree of anisotropy. For that, we propose a second modification, $\Lambda_2$ (with a fixed value $a\in (0,1]$), by the following definition,

\begin{equation}
\label{eq_Dtilda}
\Lambda_2 =
\begin{cases}
\begin{pmatrix}
a & 0 \\ 0 & 1
\end{pmatrix}\textrm{   ,on} \hspace{2mm} \partial C\\
I \hspace{13mm} \textrm{  ,otherwise.}
\end{cases};
\hspace{13mm}
A = 
\begin{cases}
V\Lambda_2V^{-1} \textrm{    ,on} \hspace{2mm} \partial C\\
I \hspace{13mm} \textrm{  ,otherwise.}
\end{cases}.
\end{equation}
In this case we also assume, for simplicity, that $\sigma \to 0$ and that the set $C$ is of bounded curvature such that $v1,v2$ align exactly with the gradient and level-set directions, respectively.
%
%where $C\subset\mathbb{R}^2$ is the set of which f is its indicator function.

% \begin{equation}
% \tilde{D} = \begin{pmatrix}
% g(|\nabla u|;K) &0 \\ 0 &1
% \end{pmatrix}; 
% \hspace{1cm}
% g(s;K) = \frac{1}{1+\left(\frac{s}{K}\right)^2}.
% \end{equation}

% Where  $g(s;K)$ is a Perona-Malik nonlinear diffusion function. 
% We can then obtain the diffusion matrix,

% \begin{equation}
% A=V\tilde{D}V^{-1}.
% \end{equation}

% As this filter generated mediocre results in our initial implementation, we considered a second filter, as suggested by Weickert with the function $c(s;K)$,

% \begin{equation}
% \tilde{D} = \begin{pmatrix}
% c(\mu_1;K) &0 \\ 0 &1
% \end{pmatrix};
% \hspace{1cm}
% c(s;K) = \Bigg\{
% \begin{array}{lc}
% 1  & s\leq 0\\
% 1-e^{-\frac{C_m}{(s/K)^m}}  &s>0
% \end{array}.
% \end{equation}

\newpage
\subsection{Regularizing Using A$^2$TV}
\subsubsection{The A$^2$TV Flow}
We denote by $u(t;x)$ the A$^2$TV flow, which is the solution of the following PDE,
\begin{equation}
\label{eq_AATVflow}
\begin{array}{lcl}
\frac{\partial u}{\partial t} = \Div_A \left( \xi^A(x) \right) &on &(0,\infty)\times\Omega \\\\
\frac{\partial u}{\partial n} = 0  &in &(0,\infty)\times\Omega \\\\
u(0;x) = f(x) &in &x\in\Omega,
\end{array}.
\end{equation}
where $\xi^A$ admits $ \argsup_{\xi^A \in C_c^\infty, \|\xi^A\|_\infty\leq 1}\int_\Omega u(x)\Div_A\xi^A(x)dx$.
%\GG{Garsmair has shown existence and uniqueness of the flow.}\SB{Why repeat? it says so in the intro...}

\subsubsection{Adaptive ROF Model}
As in the TV case, an ROF-type energy can be formulate, based on the functional $J_{A^2TV}(u)$,
\begin{equation}
\label{eq_AATVROFModel}
\mathscr{E}_{A^2TV}(u)= t\cdot J_{A^2TV}(u) + \frac{1}{2}\|u-f\|_{L_2(\Omega)}^2.
\end{equation}
In order to better understand the characteristics of this regularization, we turn to analyzing the nonlinear eigenfunctions induced by the subgradient of $A^2TV$ .
We want to show that the A$^2$TV functional can adapt itself to a wide variety of shapes, which incorporate all the eigenfunctions of the isotropic TV functional.
%, and even generalize to make almost any shape to an eigenfunction.
%Interesting enough, there is a link between the flow and the ROF model in the form of a heuristic simple equation, $\lambda = 1/t$. This link has been proven as a fact in \cite{Gilboa2016} for TV-flow in a discrete setting and should be qualitatively very similar in any dimension.
\subsection{Basic Properties of A$^2$TV}
The A$^2$TV functional is 1-homogeneous,
%This is straightforward in the strong-sense definition \eqref{eq_AATVstrongDef}, 
since for any scalar $\mu \in \mathbb{R}$ we have
%We shall show it is a 1-homogeneous functional, ,
\begin{equation}
    J_{A^2TV}(\mu u) = \int_\Omega |\nabla_A \mu u(x)|dx = |\mu| \int_\Omega |\nabla_A u(x)|dx = |\mu| J_{A^2TV}(u).
    \label{eq_1homofunctional}
\end{equation}
In addition, as for any absolutely one-homogeneous functional, the sub-gradient of the A$^2$TV, $p(u) \in\partial J_{A^2TV}(u)$ is zero-homogeneous: $p(\mu u)=p(u)$ .
\off{Which can be seen through the definition of the functional, $J_{A^2TV}(u) = \left<u,p\right>$,
\begin{equation}
    J_{A^2TV}(\mu u) = \left<\mu u,p(\mu u)\right> = \mu \left<u,p(\mu u)\right> \underbrace{=}_{J_{A^2TV}\textrm{ is }1-h} \mu J_{A^2TV}(u),
\end{equation}
after dividing by $\mu$ we get $p(\mu u) = p(u)$ which suggests indeed p is 0-homogeneous.
}
%An important property of the A$^2$TV functional is its eigenfunctions and their properties; for that reason 
The sub-gradient can be written as,
\begin{equation}
p(x) = \Div_A(\xi^A(x)), 
\label{eq_subgradientAATV}
\end{equation}
where $\xi^A(x)$ admits $ \argsup_{\xi^A \in C_c^\infty, \|\xi^A\|_\infty\leq 1}\int_\Omega u(x)\Div_A\xi^A(x)dx$.

We refer to $u$ as an eigenfunction of $J_{A^2TV}$ if it admits the following eigenvalue problem,
\begin{equation}
\label{eq_ef_AATV}
    p(u) = \lambda u \in \partial J_{A^2TV} (u),
\end{equation}
where $\lambda \in \mathbb{R}$ is the corresponding eigenvalue. 
For any $u$ not in the null-space of $J_{A^2TV}$ we have that lambda is strictly positive since from \eqref{eq_ef_AATV} we have $\langle p(u),u \rangle = \langle \lambda u, u \rangle $, yielding $\lambda = \frac{J_{A^2TV}}{\|u\|_2^2} > 0$.
%$p(u)$ is the sub-gradient and $\partial J_{A^2TV}(u)$ is the sub-differential set of the A$^2$TV functional at $u$.
%\GG{Moved this part up here.}\SB{Ok}
 For the A$^2$TV flow \eqref{eq_AATVflow}, when the initial condition $f$ is an eigenfunction, $\lambda f \in \partial J_{A^2TV}(f)$, the analytic solution is,
\begin{equation}
    u(t,x) = \Bigg\{ 
    \begin{array}{lcl}
    (1-t\lambda)f     & \textrm{for} & t \in [0,\frac{1}{\lambda}],\\
    0                 & \textrm{otherwise}.
    \end{array}.
\label{eq_AATVflowSolution}
\end{equation}
This result is shown in \cite{Gilboa2016} for all absolutely one-homogeneuous functionals.

Let us now focus on A$^2$TV as defined for sets in Eq. \eqref{eq_Dtilda} as we would like to link the A$^2$TV energy to the perimeter of a shape $C\subseteq\Omega$. If $u$ is an indicator function of $C$, it can be shown that
\begin{equation}
    J_{A^2TV}(u) =  a Per(C),
\end{equation}
where $a$ is defined in \eqref{eq_Dtilda}. This comes from the coarea formula of TV and using
$$
    J_{A^2TV}(u) = \int_\Omega |\nabla_A u(x) |dx = \int_\Omega |A(x)\nabla u(x)|dx = \int_\Omega |a\nabla u(x)|dx = a \int_\Omega |\nabla u(x)|dx,
$$
where the gradient should be understood in the distributional sense.
%\GG{Decided to simplify here and remove the detailed approximation theory arguments.}\SB{Ok.}
%
\off{
In \cite{A2TV_Properties_meyers1964hdw} Meyers at el and refined for sobolev spaces later on in \cite{A2TV_Properties_chambolle2010introduction} by Chambolle at el, we have Meyers-Serrin’s approximation Theorem which shows that BV functions may be “well” approximated with smooth functions.
Meaning, let $\Omega\subset\mathbb{R}^N$ be an open set and let $u\in BV(\Omega)$: then there exist a sequence $(u_n)_{n\geq 1}$ of functions in $C^\infty(\Omega\cap W^{1,1}(\Omega)$ such that,
\begin{enumerate}
    \item $u_n \rightarrow u \in L^1(\Omega)$.
    \item $J(u_n) = \int_\Omega|\nabla u_n(x)|dx \rightarrow J(u) = \int_\Omega|Du(x)|dx$ as $n\rightarrow\infty$.
\end{enumerate}
Therefore, by recalling that the gradient inside the set vanishes and on the edges it is the eigenvector of the eigenvalue $a$ of the adapted matrix $A$, and by using Meyers-Serrin’s approximation Theorem; for $u=\chi_C$ which has discontinuities, we shall use a series $u_n$ such that $u_n\rightarrow u$ and $J_{A^2TV}(u_n) \rightarrow J_{A^2TV}(u)$ as $n\rightarrow\infty$,
\begin{equation}
    J_{A^2TV}(u_n) = \int_\Omega |\nabla_A u_n(x) |dx = \int_\Omega |A(x)\nabla u_n(x)|dx = \int_\Omega |a\nabla u_n(x)|dx = a \int_\Omega |\nabla u_n(x)|dx = a Per(C),
\end{equation}
now, since $J_{A^2TV}(u_n) \rightarrow J_{A^2TV}(u)$ as $n\rightarrow\infty$, we have $J_{A^2TV}(u) = a Per(C)$.
Thus we see that as in the TV case, the A$^2$TV functional admits a co-area formula.
} % end off
For an eigenfunction $u$ which is an indicator function of $C$, the eigenvalue $\lambda$ can be computed by,
%\begin{equation}
%    \left<u,p(u)\right> = \left<u,\lambda u\right> \Rightarrow \lambda = %\frac{\left<u,p(u)\right>}{\left<u,u\right>} = \frac{J_{A^2TV}(u)}{\norm{u}^2} %= a\frac{Per(C)}{|C|}. 
%\end{equation}
\begin{equation}
     \lambda =  \frac{J_{A^2TV}(u)}{\norm{u}^2} = a\frac{Per(C)}{|C|}. 
\end{equation}

\subsection{Characteristics of A$^2$TV Eigenfunctions}
As in the TV case,  we can discuss the eigenfunctions induced by A$^2$TV. 
Developing a theory for the Adapted Anisotropic Total Variation  gives way to characterize the basic structures the regularizer preserves. 
%A$^2$TV regularizer is aspiring to; or, in other words, whats are the characteristics of its Eigenfunction.\\
Here again, there are two main conditions that should be fulfilled by  eigenfunctions: \emph{The convexity and the curvature conditions}.
By extending the definitions of \cite{TVFlowInRN}, we can provide part of the theory of the A$^2$TV flow and its eigenfunctions.
For this purpose, we use the matrix $A$ with  $\Lambda_2$, as defined in Eq. \eqref{eq_Dtilda}.
%We investigate the case where A$^2$TV regularization is applied to the edges whereas standard isotropic TV is applied for homogeneous regions (the rest of the image for characteristic functions).

\subsubsection{A non-convexity condition}
In order to prove a new condition related to convexity, we first introduce several definitions.
Let us set $\lambda_C^A$ as,
\begin{equation}
\label{eq_lambdaAATV}
\lambda_C^A = a\frac{P(C)}{|C|},
\end{equation}
where, $C\subset\mathbb{R}^2$ is a bounded set of finite perimeter and $a$ is the parameter of Eq. \eqref{eq_Dtilda}. Let us now define an $A$-calibrable set.

\begin{definition}[A-Calibrable Set]
\label{def_ACalibrableSets}
Let $C\subseteq\mathbb{R}^2$ be a set of finite perimeter. We say that $C$ is a $-A$-calibrable set if there exists a vector field $\xi^{A-}_C:\mathbb{R}^2\to\mathbb{R}^2$ with the following properties:
\begin{enumerate}
\item $\xi^{A-}_C\in L^2_{loc}(\mathbb{R}^2;\mathbb{R}^2)$ and  $\Div_A(\xi^{A-}_C)\in L^2_{loc}(\mathbb{R}^2)$.
\item $|\xi^{A-}_C|\leq 1$ almost everywhere in $C$.
\item $\Div_A(\xi^{A-}_C)$ is constant on $C$.
\item $\theta(\xi^{A-}_C,-D\chi_C)(x)=-1$ for $\mathscr{H}^1$-almost every $x\in\partial^*C$.
\end{enumerate}
\end{definition}
We say that $C$ is $+A$-calibrable if there exists a vector field $\xi^{A+}_C:\mathbb{R}^2\to\mathbb{R}^2$ satisfying properties (1),(2),(3) and such that $\theta(\xi^{A+}_C,-D\chi_C)(x)=1$.\\

Basically, condition (4) means that $\xi^{A-}_C$ ($\xi^{A+}_C$) is identical to the outer (inner) normalized gradient of $\chi_C$ for almost every $x\in\partial^*C$.\\

\begin{rem}
\label{xi^A_def}
Let $C\subset\mathbb{R}^2$ be a bounded set of finite perimeter. Assume that $C$ is $-A$-calibrable and that $\mathbb{R}^2\diagdown C$ is $+A$-calibrable. Define,
\begin{equation}
\label{eq:xi_A}
\xi^A_C := -\Bigg\{
\begin{array}{lcl}
\xi^{A-}_C                         & on &C\\
\xi^{A+}_{\mathbb{R}^2\diagdown C}  & on &\mathbb{R}^2\diagdown C
\end{array},
\end{equation} 
then, $\xi^A_C \in L^\infty(\mathbb{R}^2;\mathbb{R}^2)$ and $\Div_A(\xi^A_C )\in L^\infty(\mathbb{R}^2)$.\\
\end{rem} 
%
%
% The article "The Total Variation Flow in $R^N$" \citep{TVFlowInRN}, have showed that for calibrable sets (- or +) $\Omega$ and $\mathbb{R}^2\diagdown\Omega$; while on $\Omega$, the $div(\xi^-_\Omega)$ is an arbitrary constant; on $\mathbb{R}^2\diagdown\Omega$, the $div(\xi^+_{\mathbb{R}^2\diagdown\Omega})$ is a constant zero.\\
% meaning,
Here, similarly to Eq. \eqref{eq_divxi}, applying the operator $\Div_A$ yields,
 \begin{equation}
 \label{DivA_Constant}
\Div_A(\xi^A_C) := \lambda^A_C\chi_C = \Bigg\{
\begin{array}{lcl}
\lambda^A_C       & on &C\\
0  & on &\mathbb{R}^2\diagdown C
\end{array}.
\end{equation} 
%
% \textbf{The convexity condition}
We can now state the following non-convexity condition.
\begin{theorem}[The Non-Convexity Condition]
\label{Theorem_nonConvex}
Let $C\subset\mathbb{R}^2$ be a bounded set of finite perimeter which is $-A$-calibrable,  and let $\mathbb{R}^2\diagdown C$ be $+A$-calibrable. Then,
\begin{enumerate}
\item $C$ can be non-convex to some extent. 
\item $\forall D\subseteq C$, we have $\lambda^A_D \geq\frac{1}{a}\lambda^A_C$.
\end{enumerate}
\end{theorem}

\begin{proof}
Let $\xi^A_C \in L^\infty(\mathbb{R}^2;\mathbb{R}^2), \| \xi^A_C \|\leq 1$ be the vector field defined by Eq. \eqref{eq:xi_A}.
Let $D$ be a set of a finite perimeter. Using Eq. \eqref{DivA_Constant} and part 2 of Def. \ref{def_ACalibrableSets}, 
we have,\\
\[
\begin{aligned}
%Per(\Omega) = \lambda_\Omega |\Omega|
\lambda^A_C |C\cap D| 
&= \lambda^A_C\int_{\mathbb{R}^2}\chi_C\chi_D 
= \int_{\mathbb{R}^2}\lambda^A_C\chi_C\chi_D 
= \int_{\mathbb{R}^2}\Div_A(\xi^A_C )\chi_D \\
&= \int_D \Div_A(\xi^A_C)
=\underbrace{ \oint_D \xi^A_C \cdot \hat{\nu}_D}_{\substack{
\textrm{Divergence Theorem}\\ \hat{\nu}_D \textrm{ outer normal of $D$}}}
\leq \underbrace{\oint_D dl}_{\|\xi^A_C\|\leq 1,\hspace{1mm}\|\hat{\nu}_D\|=1} = Per(D).
\end{aligned}
\] 
Hence, $\forall D\supset C$, we have $aPer(C) = \lambda^A_C |C| = \lambda^A_C |C\cap D|$ which gives us $aPer(C) \leq Per(D)$. Thus $C$ can be non convex, depending on the parameter $a$ (the parameter $a$ can compensate for the non-convexity of the set $C$).
As $\forall D\subset C$ we have $\lambda^A_C |D| = \lambda^A_C |C\cap D|$, we can conclude $\lambda^A_C\leq \frac{1}{a}\lambda^A_D$.
\end{proof}\\

Basically, Theorem \ref{Theorem_nonConvex} implies that a set $C$ can be an eigenfunction, also if it is not convex, provided the anisotropy parameter $a$ is small enough.
The maximal $a$ that  maintains $C$ as an eigenfunction is,
\begin{equation}
\label{eq_AATV_a_max__NonConvexity}
a_{max} = \frac{Per(D)}{Per(C)},
\end{equation}
where $D$ is the convex hull of $C$ as it is shown in figure \ref{fig:Convex_Hull}.
This coincides precisely with the convexity measure of  Peura et al. \cite{peura1997efficiency}, Eq. \eqref{eq:conv_C}.
\begin{figure}[!h]
\centering
\includegraphics[width = 0.25\textwidth]{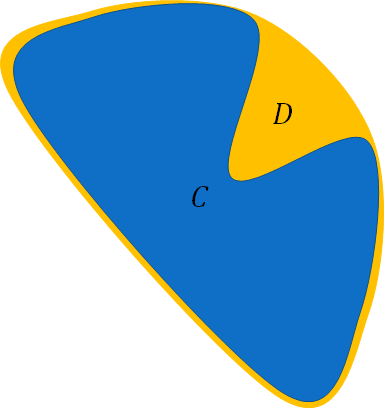}
\caption{Non-Convex shape example. Here, $D$ is the convex hull of $C$. As the perimeter of $C$ becomes larger, $a_{max}$ should be smaller to compensate for the increased degree of non-convexity of $C$, in the sense of Eq. \eqref{eq:conv_C}.}
\label{fig:Convex_Hull}
\end{figure} 
The $A$-calibrable sets include the classical calibrable sets, as stated in the following theorem.
\begin{theorem}[$A$-calibrable sets generalization]
\label{Theorem_A-Calibrable_Sets_Generalization}
Let $C\subset\mathbb{R}^2$ be a set of finite perimeter and a calibrable set. Then, it is also an $A$-calibrable set. 
\end{theorem}

\begin{proof}
Let $\xi_C:\mathbb{R}^2\to\mathbb{R}^2$ be a vector field satisfying Def. \ref{def:CalibrableSets} of the calibrable sets and $\xi^{A}_C:\mathbb{R}^2\to\mathbb{R}^2$ be a vector field satisfying Def. \ref{def_ACalibrableSets} of the $A$-calibrable sets.
Assuming $C$ is a calibrable set, we show that there exists a vector field $\xi^{A}_C$, which admits Def. \ref{def_ACalibrableSets}.

First, we consider the case where $x\in \partial C$. By definition, $\xi_C(x)=\xi_C^A(x)=\frac{Du(x)}{|Du(x)|}$ (where $Du$ is the gradient of $u$ in the distributional sense). 
%Moreover, for A$^2$TV we have the adapted anisotropic matrix $A$, which $\frac{Du}{|Du|}$ and $a$ are it's eigenvector and eigenvalue respectivly. Then,
Since at each point $x$, $\frac{Du(x)}{|Du(x)|}$ is an eigenvector of $A(x)$ with eigenvalue $a$ we have $A\xi_C^A = a\xi_C^A = a\xi_C$, therefore $\forall x\in \partial C$,
%\[
 %   \xi_C^A = \xi_C \Rightarrow A\xi_C^A = a\xi_C^A = a\xi_C \Rightarrow \Div A\xi_C^A = \Div_A\xi_C^A = \Div a\xi_C = a\lambda_C^{TV}\chi_C = \lambda_C^{A^2TV}\chi_C 
%\]
\[
    \Div_A\xi_C^A(x) = \Div A\xi_C^A(x) =  \Div a\xi_C(x) = a\lambda_C^{TV}\chi_C(x) = \lambda_C^{A^2TV}\chi_C(x).
\]
Next, we consider the case $x\in \mathbb{R}^2 \diagdown \partial C$.
%\GG{I think this should be the whole domain except $\partial C$,
%$x\in \mathbb{R}^2 \diagdown \partial C$ }
Here $A=I$ and $\Div = \Div_A$, thus,
\[
    \Div \xi_C = \lambda_C^{TV}\chi_C \Rightarrow a\lambda_C^{TV}\chi_C = \lambda_C^{A^2TV}\chi_C = a\Div \xi_C = \Div a\xi_C = \Div_A a\xi_C.
\]
By assigning $\xi_C^A(x)  = a\xi_C(x)$,  $\forall x\in \mathbb{R}^2 \diagdown \partial C$ , we get $\Div_A \xi_C^A = \lambda_C^{A^2TV}\chi_C$. 
Therefore, we obtain a vector field $\xi_C^A$, defined in the entire domain, which fully admits Def. \ref{def_ACalibrableSets}.
\end{proof}

%It is worth mentining that the original divergence operator work on $\xi_C$ in the TV case and $A\xi_C^A$ in the A$^2$TV case, which for the calibrable set results into $\xi_C^A = a\xi_C$ (inside $C$, $\xi_C^A  = a\xi_C$ and $A=I$, and on $\partial C$,  $\xi_C^A  = \xi_C$  and $A$ has an eigenvector-eigenvalue pair of $\frac{Du}{|Du|}$ and $a$ which diminish the divergence operand by a factor of $a$).
%convexity & curvature conjectures
\subsubsection{The curvature condition}
We attempt to reach a bound on the maximal curvature admitted by  eigenfunctions of A$^2$TV, analogous to \eqref{eq:max_curvature_tv}.
The theoretical analysis here is highly involved and is still under investigation.
Current known theory to provide such bounds relies on the convexity property of the set and thus cannot be applied here.
At this point, we can give some intuition on how this bound should be characterized and present compelling numerical evidence.

%Due to the complexity of the subject (most if not all proofs rely the convexity property which does not apply here), at this point, we can only give a hypothesis relating the maximal curvature on the set boundary to the degree of anisotropy (the parameter $a$). 

We would like to find the relation of the maximal curvature on the set boundary to the degree of anisotropy (the parameter $a \in (0,1]$). Let us first discuss the extreme cases. For $a=1$ we get isotropic TV, therefore the bound should be as in  \eqref{eq:max_curvature_tv}. As $a$ decreases, the regularization diminishes and we expect the bound to grow as $a \to 0$.

%We would like to walk with the reader through the process of understanding the connection and possibility of our claim.

We now turn to examining more closely the dual variable $\xi^A$.
We would like to analyze the application of the classical divergence, therefore we define
\begin{equation}
 \tilde{\xi^A} := A\xi^A, 
 \label{eq_tilde_xi^A}
\end{equation}
where $A$ is the adapted anisotropic matrix.
In Fig. \ref{fig:Curvature_Conjecture_tildeXiA_vs_AXiA}, 
one can see the differences between the vector fields $\tilde{\xi^A}$ and $\xi^A$. Note especially the smoother behavior of $\tilde{\xi^A}$ near edges.
%definitions are the smoothness of the first and the peaks at the edges at the last. 
%The smoothness of $\tilde{\xi^A}$ is an important property since the classical divergence is applied on it, while the equal absolute one spikes are a basic property of the A-calibrable sets of the A$^2$TV as described in part 4 of Def. \ref{def_ACalibrableSets}.

\begin{figure}[H]
\centering
\includegraphics[width = 0.8\textwidth]{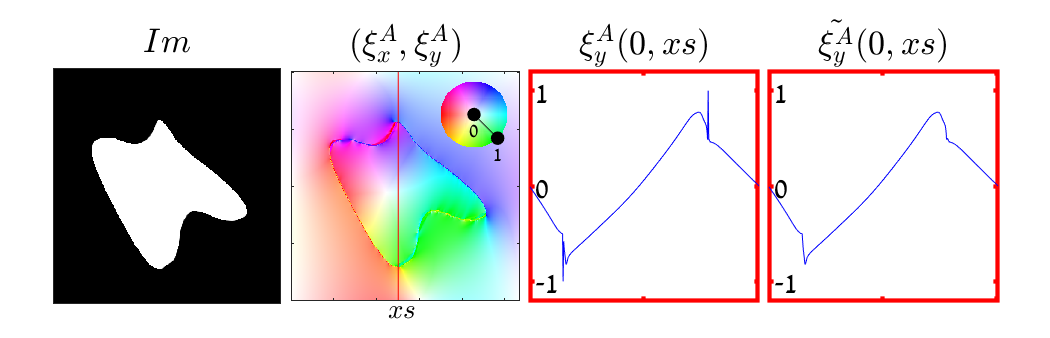}
\caption[$\tilde{\xi^A}$ vs $\xi^A$ comparison]{Comparison between the vector fields $\xi^A$ (defined by Eq. \eqref{eq_AATVstrongDef}) and $\tilde{\xi^A}$, Eq. \eqref{eq_tilde_xi^A}, (a=0.5). From left: indicator function of a nonconvex set $C$, color coded vector field $\xi^A$ and the cross section vertical line (in red), cross sections of $\xi^A$ and $\tilde{\xi^A}$. Whereas $\xi^A$ is not smooth near the edge, $\tilde{\xi^A}$ is smooth.}
\label{fig:Curvature_Conjecture_tildeXiA_vs_AXiA}
\end{figure}

Let us compare the TV definition \ref{def:CalibrableSets} and the A$^2$TV definition \ref{def_ACalibrableSets}. For eigenfunctions, $\xi$ for TV and $\tilde{\xi^A}$ for A$^2$TV, should have a constant divergence. We denote both vector fields by $\xi$. In two dimensions, a vector field has a constant divergence when its components are of the form,
\begin{equation}
\Big\{
\begin{array}{lcl}
\xi_x  =&  R(x)R'(y) + t_1x \\
\xi_y  =& -R'(x)R(y) + t_2y, 
\end{array}
\end{equation} 
for some differential function $R(\cdot)$. Clearly, we get by construction a constant divergence: $\textrm{div}(\xi) = t_1+ t_2$. For total variation the function $R(\cdot)$ is identically zero (for  a disk of radius $r$ we get $R(\cdot)=0, t_1=t_2=\frac{1}{r}$). However, our numerical observations indicate it is a more complex function for A$^2$TV. For example, in the case of ellipses, $\tilde{\xi^A}$ has a structure similar to a third degree polynomial (as can be seen in Fig \ref{fig:Curvature_Conjecture_6__2_AATV_Ellipses_Flows_Fixed_Ratio}).
%\GG{Below for instance, you need to explain why you show $\xi^A$ and what do you expect from it, for instance - at the edge it should not be differentiable, only $A \xi^A$ is..}

%\GG{Put experiment we talked about: Constant $a$, ellipse minor-axis changing, all results are eigenfunctions, show plot of major-axis cross-section - how the "bump" grows.}\\
Several numerical experiments were conducted using ellipse shapes (as in Fig. \ref{fig:Curvature_Conjecture_Ellipse}) with various eccentricities and A$^2$TV regularization with various degrees of anisotropy parameter ($a$). An ellipse shape is used in order to decouple the convexity and the curvature measures. Moreover, it is simple to estimate numerically its curvature. 

\begin{figure}[H]
\centering
\includegraphics[width = 0.6\textwidth]{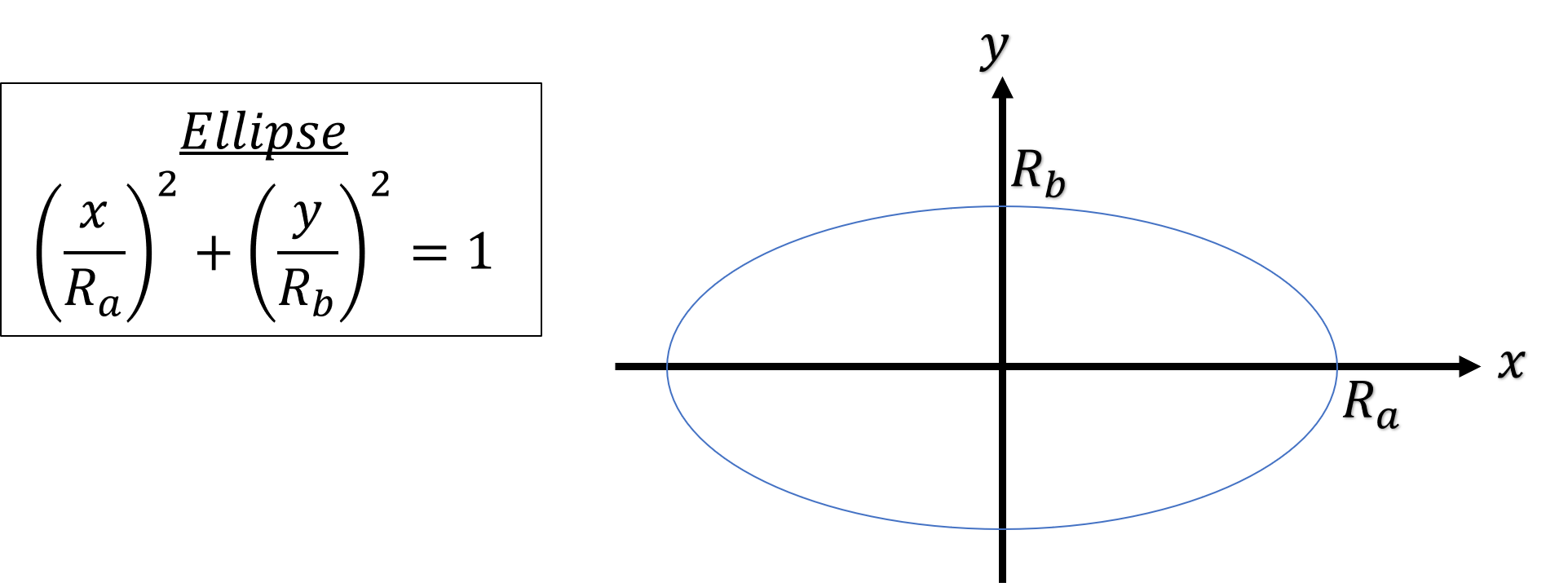}
\caption[Ellipse]{Basic ellipse shape}
\label{fig:Curvature_Conjecture_Ellipse}
\end{figure}
For an ellipse as in Fig. \ref{fig:Curvature_Conjecture_Ellipse},  $R_a\geq R_b$, the maximum curvature, located at $(x,y)=(R_a,0)$, is,
\begin{equation}
\label{eq_ellipse_kappa}
    \kappa_{max} = R_a/R_b^2.
\end{equation}

The first experiment is done by applying A$^2$TV on a single ellipse with various values of the anisotropy parameter $a$. The resultant $\tilde{\xi^A}$ cross-sections is shown  in Fig \ref{fig:Curvature_Conjecture_6__2_AATV_Ellipses_Flows_Fixed_Ratio}.
%\GG{REVISED HERE (see if you agree).
One can observe that the minor axis (of $\tilde{\xi^A_y}$) has a linear solution just as in the TV case (due to its low curvature). The major axis, on the other hand, is not a linear function. The curvature on the boundary is maximal and $\tilde{\xi^A}$ is not monotone inside the set anymore.
For $a>0.5$ the anisotropy of A$^2$TV is not sufficient and one does not obtain an eigenfunction, this can be seen by the clipping of $\tilde{\xi^A_x}$ near the value of $1$. For $a\le 0.5$ we get an eigenfunction with $\Div(\tilde{\xi^A}) = const$, here $\tilde{\xi^A_x}$ resembles a 3rd degree polynomial.
%anisotropy parameter $a$ of the A$^2$TV flow is trying to compensate on. In this example, 
%One can see that as $a$ become smaller, the deviation from a linear solution is lager, compensating for the high curvature. 
%
%In the same time the dual variable $\xi$ becomes more 'natural', undistorted and without clipping.    

\begin{figure}[H]
\centering
\includegraphics[width = 0.9\textwidth]{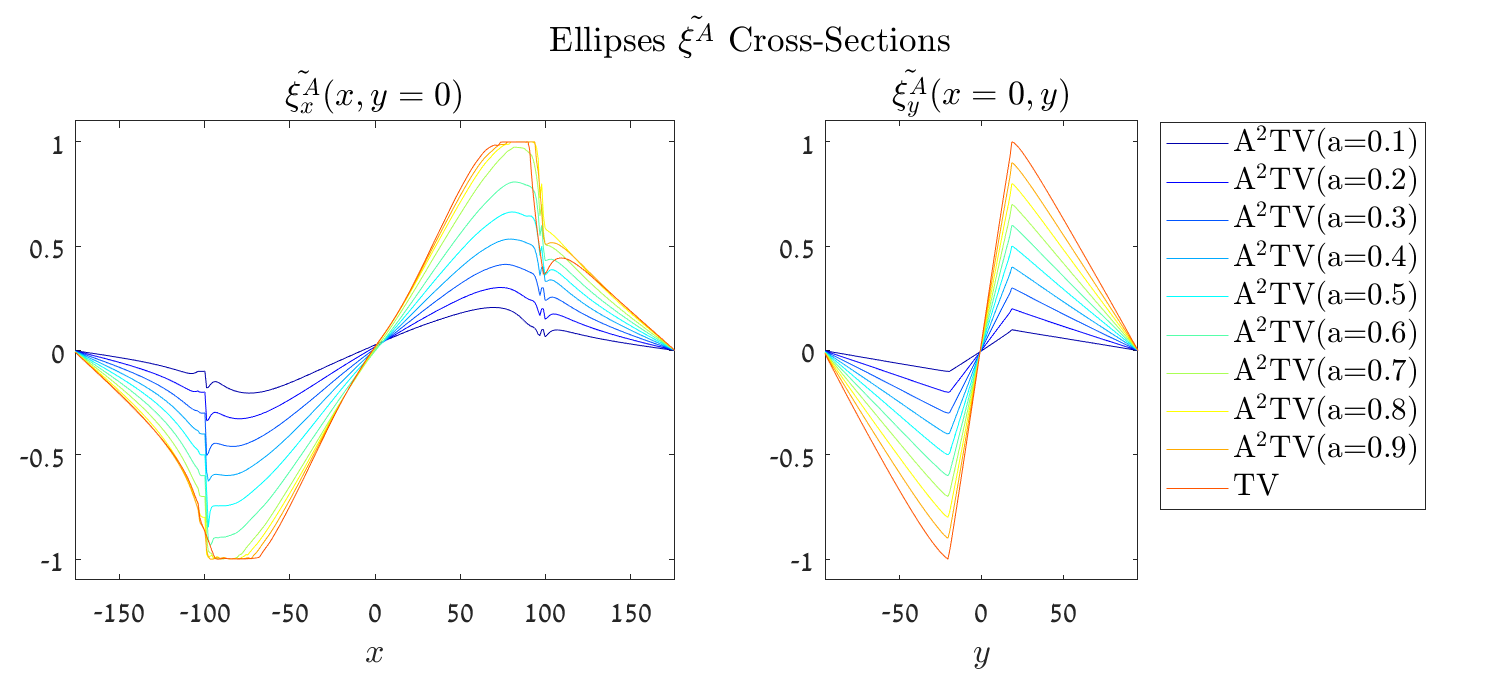}
\caption[A$^2$TV Ellipse $\tilde{\xi^A}$ Cross-Sections.]{A$^2$TV Ellipse $\tilde{\xi^A}$ Cross-Sections - A numerical experiment to asses the behavior of $\tilde{\xi^A}$ as a function of $a$. The parameters of the ellipse are: $Ra = 100, Rb = 20$.
%\GG{Seems like you are showing $\xi^A$ and not $\xi = A \xi^A$. Here you have this "kink" on the edge and the function is not smooth, unlike what you have in the second row of Fig. 6.2. I think showing $A \xi^A$ is better, this is what the divergence is operating on and can be compared directly to TV, as we say in the text.} \SB{On the other hand, $\xi^A$ is the one talked about in the theory... }
%\GG{If you want to keep it this way, we need to be consistent in the text, not say we analyze only $\xi$, explain why you choose to plot $\xi^A$ also everywhere here write $\xi^A_x$, $\xi^A_y$ in the figure, etc.}
}
\label{fig:Curvature_Conjecture_6__2_AATV_Ellipses_Flows_Fixed_Ratio}
\end{figure} 

In Fig. \ref{fig:Curvature_Conjecture_6__3_AATV_Ellipses_Flows_Fixed_a} we present results of another numerical experiment. Now, a fixed anisotropic parameter $a=0.5$ is used and ellipses with various radii ratios ($r = \frac{R_b}{R_a}$) are tested. When the eccentricity of the ellipse is low  ($\frac{R_b}{R_a}$ closer to 1) there is a linear profile for both $\tilde{\xi^A}$ axes. Actually these ellipses are  eigenfunction also of TV (ratios of $r\in(0.6,1]$). As the eccentricity grows with $r\in(0.15,0.5]$, we obtain eigenfunctions unique to A$^2$TV (with $a=0.5$). Finally, we have a very narrow ellipse, $r = 0.1$, which is not an eigenfunction, where the vector field is clipped.

\begin{figure}[H]
\centering
\includegraphics[width = 0.9\textwidth]{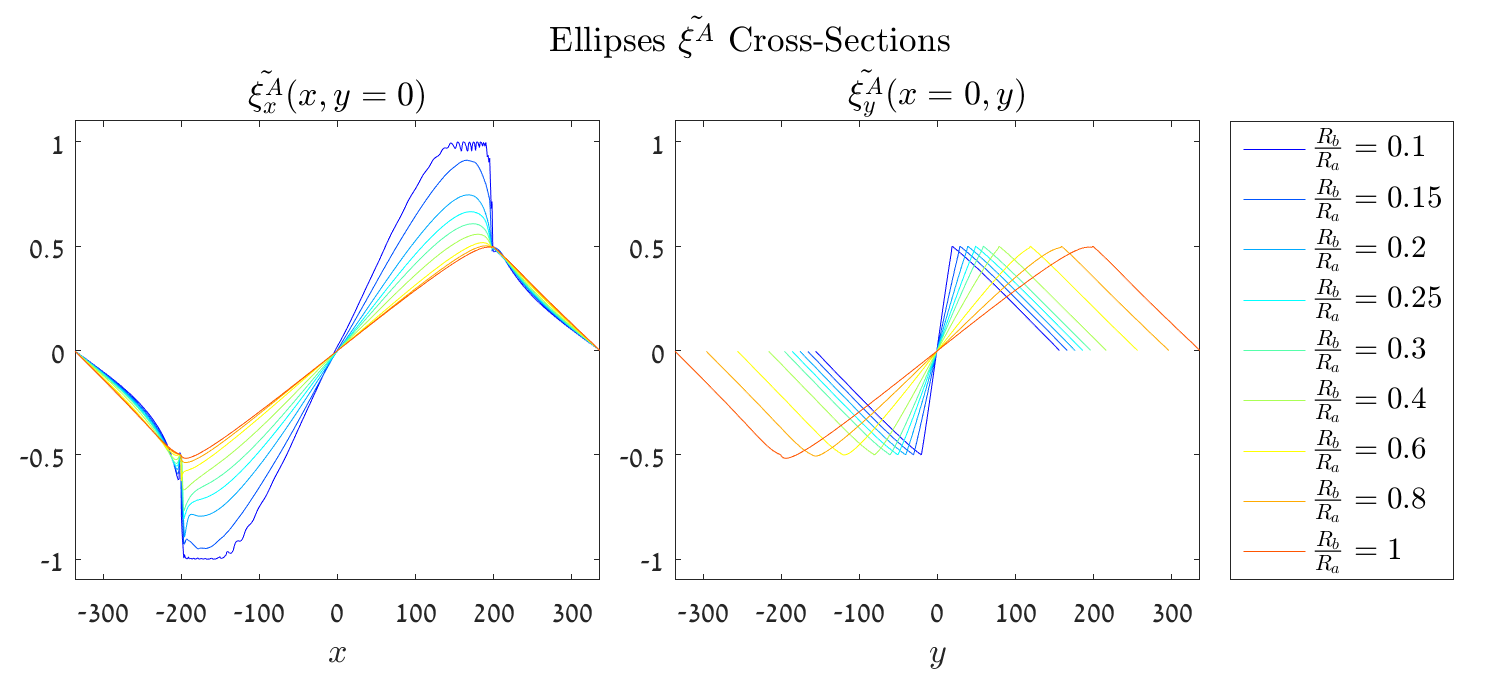}
\caption[A$^2$TV ellipse $\tilde{\xi^A}$ cross-sections]{A$^2$TV Ellipse $\tilde{\xi^A}$ Cross-Sections - An experiment to asses the behavior of $\tilde{\xi^A}$ as a function of the ratio $r = \frac{R_b}{R_a}$. Two cross-sections of $\tilde{\xi^A}$ along the axes are shown. ($a=0.5$, $R_a = 200$).
%\GG{Same comment as in previous figure, showing $A\xi^A$ is better.}
}
\label{fig:Curvature_Conjecture_6__3_AATV_Ellipses_Flows_Fixed_a}
\end{figure} 

These experiments show a distinct systematic trend relating the parameter $a$ to the set of ellipse eigenfuncions, which is summarized in Fig.\ref{fig:Curvature_Conjecture_Ellipse_Family}. As  $a$ becomes smaller, ellipses of higher eccentricity  (and curvature) become admissible eigenfunctions.

\begin{figure}[H]
\centering
\includegraphics[width = 0.7\textwidth]{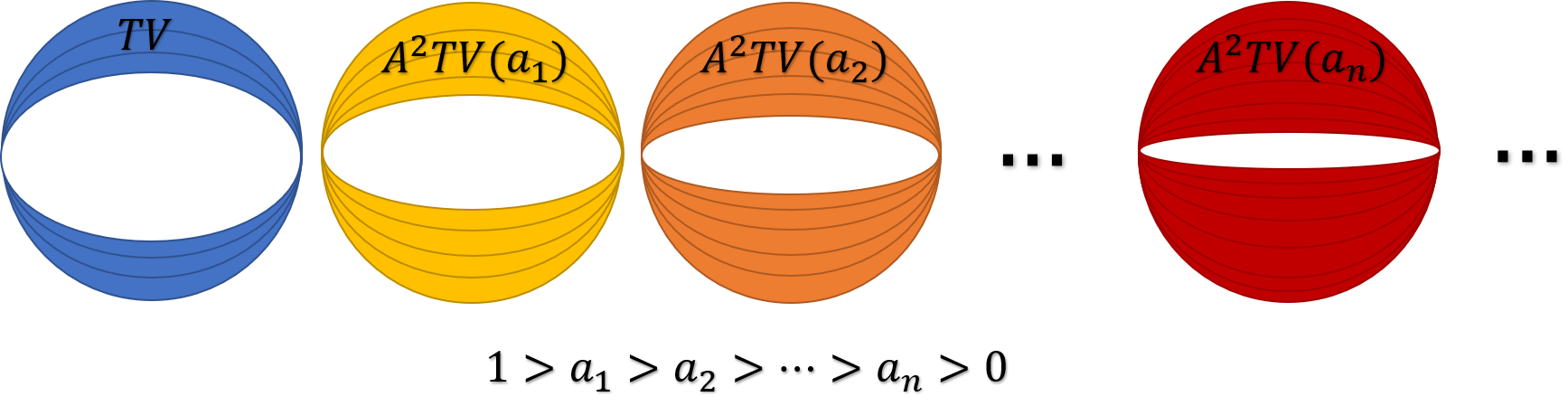}
\caption[A$^2$TV Ellipse Eigenfunction Sets.]{{\bf A$^2$TV Ellipse Eigenfunction Sets.} A visualization of the possible eigenfunction sets (colored region) of ellipses for the A$^2$TV functional as the parameter $a$ decreases (from left to right).}
\label{fig:Curvature_Conjecture_Ellipse_Family}
\end{figure} 

In order to estimate numerically a curvature condition for ellipses we performed the following experiment, with the results depicted in  Fig. \ref{fig:Curvature_Conjecture_6_AATV_Ellipses_Flows}. Multiple A$^2$TV flows (with different $a$ parameter) were applied to ellipses of different eccentricities ($R_b/R_a \in (0.2,1]$). 

\begin{figure}[H]
\centering
\includegraphics[width = 0.9\textwidth]{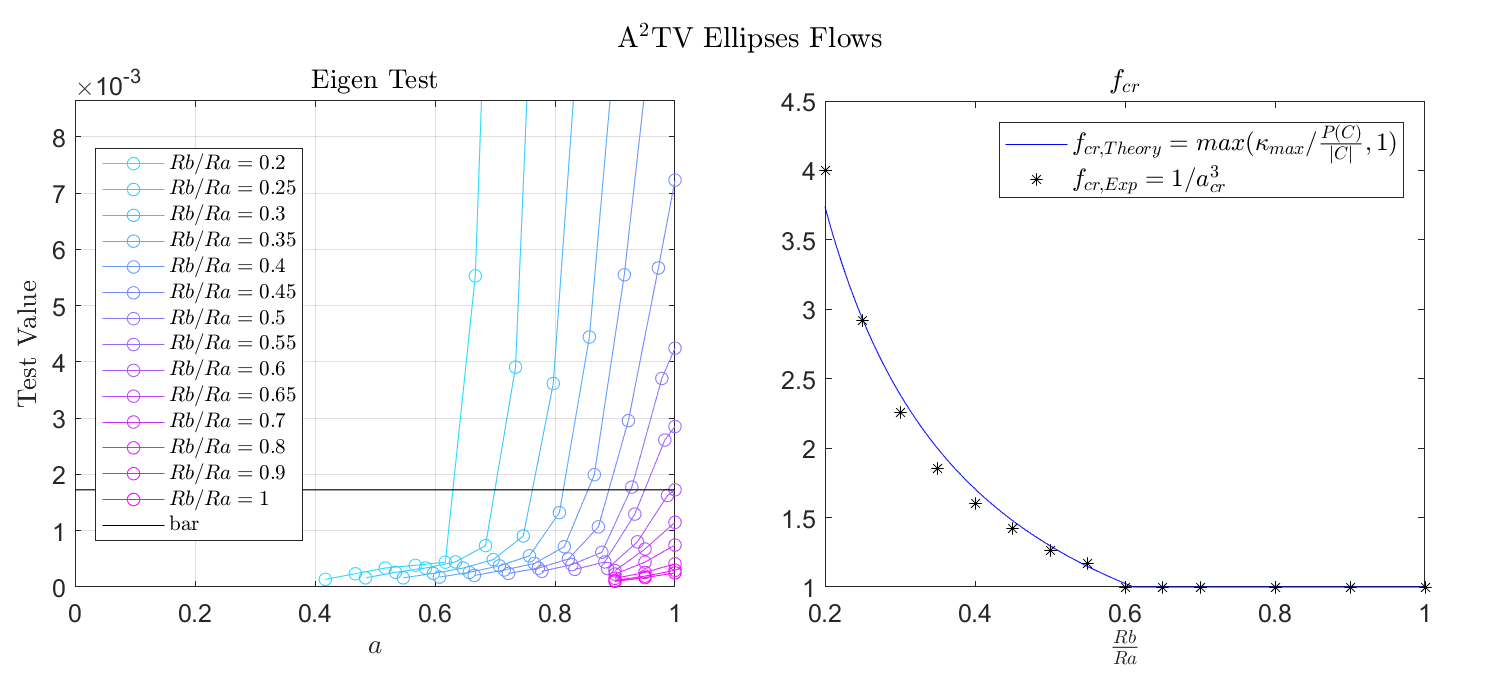}
\caption[A$^2$TV Ellipses Flows.]{Regularization by A$^2$TV of ellipses - a numerical experiment to asses the behavior of the function $f_{cr}(a)$. Here, numerous A$^2$TV tests were performed with anisotropic parameter $a\in (0.4,1]$, radii ratios of $R_b/R_a = (0.2,1]$, $R_a = 200$.}
\label{fig:Curvature_Conjecture_6_AATV_Ellipses_Flows}
\end{figure}

%\GG{Explain why ellipse is used: hard to compute numerically the curvature of a general shape. Also we want to decouple the convexity and curvature measures, so we chose a convex shape. Put the formula of the maximal curvature of an ellipse with major semi-axis $a$ and minor semi-axis $b$ (is this what you call radii $R_a$, $R_b$?). Put a plot of an ellipse with $a$, $b$ drawn, take all definitions and notations from wikepedia https://en.wikipedia.org/wiki/Ellipse}

In Fig. \ref{fig:Curvature_Conjecture_6_AATV_Ellipses_Flows} we can see two plots. The left plot shows the eigenfunction test scores (lower is better) based on the following formula,
\begin{equation}
    T(u) = \frac{u(0,0) - u(R_a-\epsilon,0)}{u(0,0)},
\end{equation}
where $u$ is an ellipse image after applying a A$^2$TV flow at $t = \frac{1}{10\lambda^A}$ where $\lambda^A$ is taken from Eq. \eqref{eq_lambdaAATV}. The test score, typically in the range $T(u)\in [0,1]$, was chosen due to  its numerical robustness, as oppose to indicators which take into account the entire image, and are biased by numerical errors. We rely on the fundamental property of indicator set eigenfunctions which retain a constant value within the set. The time of the flow chosen for the computation is 10\% of the approximated extinction time. The value of $\epsilon = 0.05 \cdot R_a$ was used (for $R_a=200$, $\epsilon = 10$).
%\GG{Put the value of $\epsilon$. }\SB{Done}
%\GG{What is this test score? Define precisely with an equation.}\SB{Done.}

%\GG{OK. Find a nice letter instead of "TestScore" it's actually EF-Test-Score. Explain why this is chosen, actually the deviation from a constant value of the flow after some time. Explain why it is rather robust numerically, what is the range of the score (for ef and gettring farther from it), also why you chose the time $1/(10 \lambda)$.}\SB{Done.}
On the left, it can be observed that as the  ratio tends to 1 (towards a disk), the shape attains a better eigenfunction score.
The bar with the value $0.0017$ was chosen so that the ratio 0.6 is the last one to be an eigenfunction for the TV functional (which is the theoretical threshold for TV).
From the graph on the left, critical thresholds for $a$, $a_{cr}$ values, were extracted with the rule of choosing $a$ which gives the same test score $T(u)$ for each ellipse ratio.
The right plot shows the theoretical $f_{cr,Theory}$, which represents the actual bound A$^2$TV needs to compensate by, 
\begin{equation}
    f_{cr,Theory} = \max(\kappa_{max}/\frac{P(C)}{|C|},1),
\end{equation}
where the set $C$ represent an ellipse with a major axis $R_a$ and a minor axis $R_b$ as shown in Fig. \ref{fig:Curvature_Conjecture_Ellipse}. For each ellipse ratio, we took the critical $a_{cr}$ and applied it as a function of $f_{cr,Exp}$,

$$f_{cr,Exp}(a_{cr}) = \frac{1}{a_{cr}^3}.$$

%The right plot shows the theoretical ratio between the maximal curvature of each ellipse and the ratio of its perimeter and area. For each ellipse ratio, we took the critical $a_{cr}$ and applied it as a function of $f_{cr}(a_{cr}) = \frac{1}{a_{cr}^3}$.
%\GG{Explain $f_{cr}$, put equation.}\SB{Done.}
%\GG{Good but still not completely clear. The whole paragraph above should be as clear and accurate as possible, as this is the basis of the conjecture. We need to have precise definitions (in equations, not words) on $a_cr$, the threshold selected of $T(u)$, (the black line on the left plot) and precisely the computations you do to have the right plot (fitting line and point).
%$f_{cr}$ in the plot is a function of ($R_b/R_a$ and not of $a_{cr}$, so this should be noted.)
%}\SB{Done my best, I dont have a function for $a_{cr}$,and as for your remark that $f_{cr}$ is a function of the ratio of the radii, I hope I managed to explain that $f_{cr}$ is calculated per ratio as stated..}

Both images in Fig.\ref{fig:Curvature_Conjecture_6_AATV_Ellipses_Flows} present a strong indication that the upper bound dependence on the parameter $a$ is $\frac{1}{a^3}$. Note that this is the case for ellipses, it is hard to validate if curvature and convexity/shape structures can be completely decoupled.
%\GG{For ellipses (or convex shapes), we cannot be sure the curvature and convexity/shape structure can be decoupled in all cases..}\SB{Ok.. Agreed.. unfortunally...}
% 
We summarize this relation by the following conjecture,
%\GG{It may be wise to limit the scope of the conjecture to ellipses (or convex shapes).}\SB{Updated in the conj}
\begin{conj}
Let $C \subset \mathbb{R}^2$ be a bounded ellipse  for which its indicator function is an eigenfunction in the sense of \eqref{eq_ef_AATV}, with $\kappa(x)$ the curvature at point $x\in \partial^*C$. Then,
\begin{equation}
\label{eq_AATV_curvature_upper_bound}
\max_{\forall x\in \partial^*C}\kappa(x)\leq \frac{\lambda_C^A}{a^4}=\frac{1}{a^3}\frac{P(C)}{|C|}.
\end{equation}
%Where $a\in (0,1)$ is the new eigenvalue in matrix $A$.
\end{conj}\\ 
Another way to phrase it, is that the maximal $a$, for which an indicator function of an ellipse $C$ is an eigenfunction, is,
\begin{equation}
\label{eq_AATV_a_max__Curvature}
a_{max} = \min\left(\sqrt[3]{\frac{Per(C)/|C| }{\max\{\kappa_{\partial C}\}}},1\right).
\end{equation}

\paragraph{Choosing $a$ for the set $C$}
%\GG{This is again $a_{max}$, not $a_{min}$, only here we account for both considerations of convexity and curvature, maybe we can leave only \eqref{eq_AATV_a_max} and remove \eqref{eq_AATV_a_max__Curvature}.}\SB{fixed. & I rather keep it gradual}
If one would want a general set $C$ to be well preserved (approximately an eigenfunction) by the regularizer A$^2$TV, using Eq. \eqref{eq_AATV_a_max__NonConvexity} and Eq. \eqref{eq_AATV_a_max__Curvature}, we thus suggest to select the following $a$ parameter,
\begin{equation}
\label{eq_AATV_a_max}
a_{max} = \min\left(\frac{Per(D)}{Per(C)}, \sqrt[3]{\frac{Per(C)/|C| }{\max\kappa_{\partial C}}}\right),
\end{equation}
where $D$ is the convex hull of $C$. 
%Here we added the second minimum function because we want to limit $a$ from above to 1.
%\GG{$Per(D)/Per(C) \le 1$ by definition, the 3rd term "1" is not needed.}\SB{You've got a point.. and it seems to generalize better the TV by adding one more term..}

\section{The Spectral Framework}
\label{sec:spectral}

The spectral transform is based on a flow of a 1-homogeneous functional $J(u)$ similar to the TV and the  A$^2$TV flow described above. The general gradient flow is
%, for which $J(\lambda u) = |\lambda|J(u)$ holds for all $\lambda \in \mathbb{R}$, 
\begin{equation}
\label{eq_Spectral_Framework_Flow}
u_t(t,x) = -p(t,x), \hspace{1mm} u(0,x) = f(x), \hspace{1mm} p(t,x)\in\partial_u J(u),
\end{equation}
where $f$ is assumed to be of zero mean, $t\in (0,\infty), \hspace{2mm}x\in\Omega$ and Neumann boundary conditions are used. 

The spectral transform is analogue to the Fourier transform in the sense of the transition from one domain to another, a spectral response which describes the energy at each scale ('frequency') and a filtering mechanism which enables  to diminish or amplify each scale band. With that analogy in mind, the transform is defined by
\begin{equation}
\label{eq_Spectral_Framework_Phi_Transform}
\phi(t)=u_{tt}t.
\end{equation}
With the inverse transform,
\begin{equation}
\label{eq_Spectral_Framework_inverse_Phi_Transform}
\hat{f}=\int_0^\infty\phi(t)dt
\end{equation}
The spectral response which is used here for visualization of the active scales is 
\begin{equation}
\label{eq_Spectral_Framework_Spectrum}
S(t) = \int_\Omega|\phi(t;x)|dx.
\end{equation}
It is worth mentioning there are several definitions to the spectrum response curve; among them is $S(t)=\langle f,\phi(t)\rangle$ which holds a Perseval-type relation described in \cite{Gilboa_spectv_SIAM_2014}.
Based on this framework, one can obtain a filtered responses using a filter $H(t)$:
\begin{equation}
\label{eq_Spectral_Framework_Filtering}
\phi_H(t;x) = H(t)\phi(t;x), \hspace{2mm} f_H(x) = \int_0^\infty\phi_H(t;x)dt,
\end{equation}
where $H(t)$ can be among other options, a low-pass filter (LPF) $H(t) = H_{LPF}^{t_c}$, with the following filtered response,
\begin{equation}
    H_{LPF}^{t_c} = \bigg\{
    \begin{array}{lc}
      1  , & t\geq t_c\\
      0  , & otherwise
    \end{array},
    \hspace{1cm}
    f_{H}(x) = \int_{t_c}^\infty \phi(t,x)dt,
\end{equation}
a LPF complement's high-pass filter (HPF) $H(t) = H_{HPF}^{t_c}$
\begin{equation}
    H_{HPF}^{t_c} = \bigg\{
    \begin{array}{lc}
      1  , & t\leq t_c\\
      0  , & otherwise
    \end{array},
    \hspace{1cm}
    f_{H}(x) = \int_0^{t_c} \phi(t,x)dt,
\end{equation}
and a band-pass filter (BPF) $H(t) = H_{BPF}^{[t_1,t_2]}$
\begin{equation}
\label{eq_BPF}
    H_{BPF}^{t_1,t_2} = \bigg\{
    \begin{array}{lc}
      1  , & t\in [t_1,t_2]\\
      0  , & otherwise
    \end{array},
    \hspace{1cm}
    f_{H}(x) = \int_{t_1}^{t_2} \phi(t,x)dt,
\end{equation}
%\GG{Define here more specifically a band-pass filter, $H(t)=H_{BP}^{[t_1,t_2]}(t)$, where $H_{BP}^{[t_1,t_2]}=1$ in the range $[t_1,t_2]$ and 0 otherwise. Then the filtered result becomes:
%$ f_{H}(x) = \int_{t_1}^{t_2}\phi(t,x)dt$ (put as an equation), later refer to this equation in the figures you show spectral results.}

%
For $f$ being an eigenfunction with an associated eigenvalue $\lambda$, the transform, defined by Eq. \eqref{eq_Spectral_Framework_Phi_Transform}, becomes
%is given in Eq. \eqref{eq_AATVflowSolution}, With the transform being,
\begin{equation}
\label{eq_Spectral_Framework_Phi_Eigen_Deltas}
\phi(t)=\delta\left(t-\frac{1}{\lambda}\right)f.
\end{equation}

Here we use the definition of $A$ for A$^2$TV which applies for general images, Eq. \eqref{Eq_Wickert_D_tilde}.
\off{
In order to apply the spectral framework to the TV functional or the A$^2$TV one, we shall recall they are 1-homogeneous functionals as described in Eq. \eqref{eq_1homofunctional} and by using the sub-gradients described in Eq. \eqref{eq_subgradientTV} and Eq. \eqref{eq_subgradientAATV} with the flow described in Eq. \eqref{eq_Spectral_Framework_Flow} we could create the spectral transform.
\GG{==================}

\GG{This should be put in Section 3.5 (properties), for TV, we should comment this is well known.}
\begin{equation}
J_{TV}(\lambda u) = \int_\Omega |\nabla \lambda u(x)|dx = |\lambda| \int_\Omega |\nabla u(x)|dx = |\lambda| J_{TV}(u) .
\end{equation}
Now, we shall use the flow described in Eq. \eqref{eq_TVflow} where the sub-gradient $p$ in the spectral framework is,
\begin{equation}
p(t;x) = \Div(\xi(t;x))
\end{equation}
Where $\xi(t;x)$ admits $ \argsup_{\xi \in C_c^\infty, \|\xi\|_\infty\leq 1}\int_\Omega u(t;x)\Div\xi(t;x)dx$.

Similarly, for the $A^2TV$ case, we shall show its a 1-homogeneous functional,
\begin{equation}
J_{A^2TV}(\lambda u) = \int_\Omega |\nabla_A \lambda u(x)|dx = |\lambda| \int_\Omega |\nabla_A u(x)|dx = |\lambda| J_{A^2TV}(u).
\end{equation}
With the appropriate flow described in Eq. \eqref{eq_AATVflow} and the following sub-gradient $p$,
\begin{equation}
p(t;x) = \Div_A(\xi^A(t;x)). 
\end{equation}
Where $\xi^A(t;x)$ admits $ \argsup_{\xi^A \in C_c^\infty, \|\xi^A\|_\infty\leq 1}\int_\Omega u(t;x)\Div_A\xi^A(t;x)dx$.

\GG{================== until here}
}
In Fig \ref{fig:Spectrum_Illustration}  we show an example of $\phi$ bands of an image, which are narrow band-pass filters, in the sense of \eqref{eq_BPF}, of increasing scales. We compare spectral TV and spectral A$^2$TV (using $k=1$ in Eq. \eqref{Eq_Wickert_D_tilde}).
%\GG{ADD $a=XX$}\SB{Done.}). 
One can see that the main difference between the two functionals is that larger scales (we term ``objects'' here) have more distinct structures for A$^2$TV, preserving better the shapes of the original image.
Standard TV highly regularizes the large shapes, turning them into simple blobs.
\begin{figure}[H]
\centering
\includegraphics[width = 0.8\textwidth]{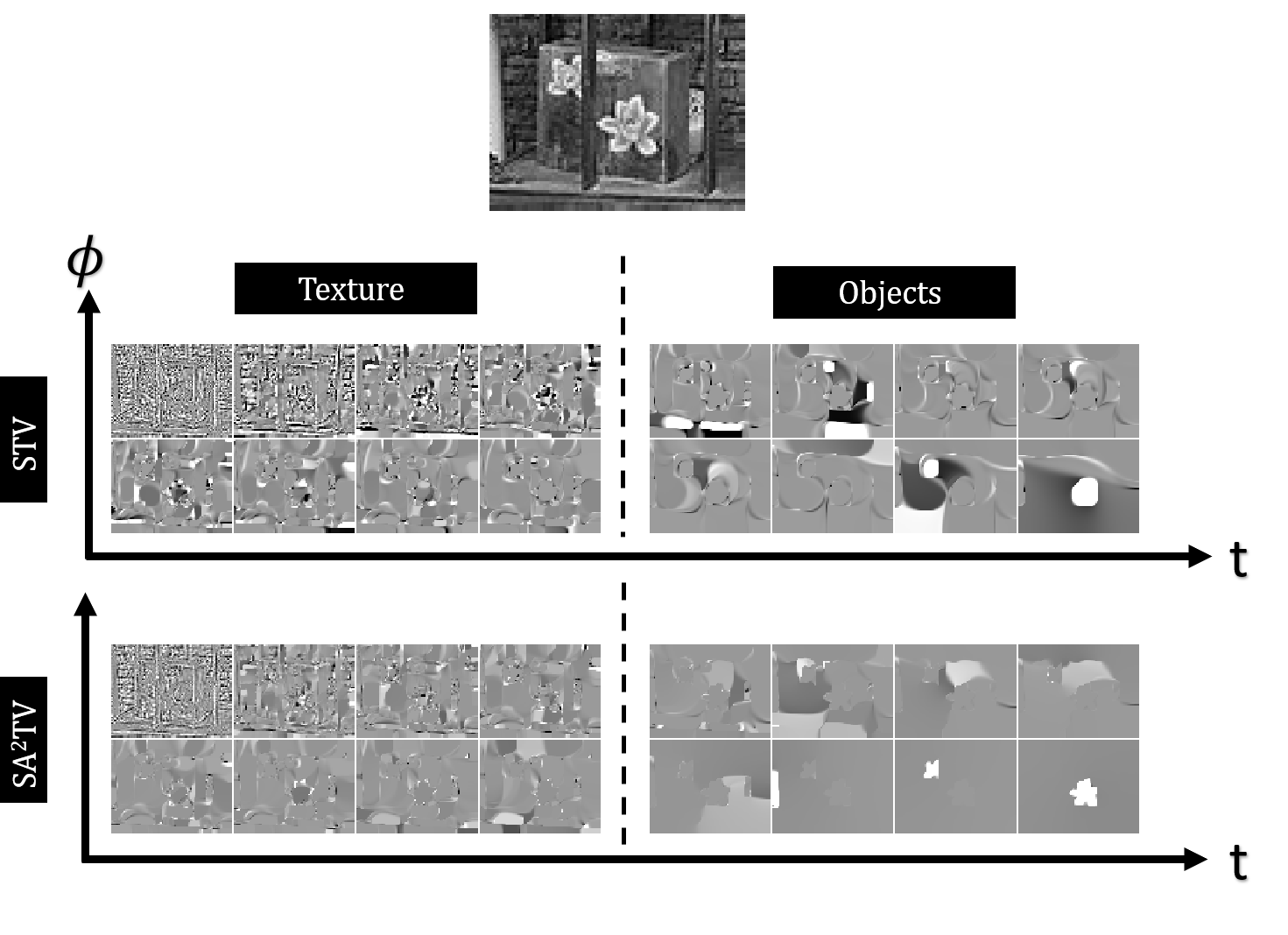}
\caption{
Spectral $\phi$ bands comparison of  spectral-TV and spectral-A$^2$TV.  The bands are divided into two main scale bands. On the left the textures of the image are shown; On the right, the macro textures, or objects. }
\label{fig:Spectrum_Illustration}
\end{figure}

% \subsection{TV Spectral Framework}
% \GG{Repetition of previous page.
% First prove (one-line) A$^2$TV is positively-one homogeneous.
% You can then directly go to the AATV spectral stuff, say because its positively-one homogeneous you can apply the entire
% theory of SIIMS-16 paper.
% }

% \subsubsection{Formulation}
% \subsubsection{Examples}
% \subsection{A$^2$TV Spectral Framework}
% \subsubsection{Formulation}

% \subsubsection{Examples}
% \GG{Bla bla} etc...
\section{Numerical Schemes}
\label{sec:numerical}
One can use the anisotropic operators for gradient and divergence, defined in Eq. \eqref{eq_GradA} and Eq. \eqref{eq_DivA} to adapt and implement the flow and minimization numerically. This can be applied to almost any minimizing method like Chambolle's projection algorithm \cite{chambolle2004algorithm}, the primal dual method of Chambolle-Pock \cite{ChambollePock_PrimalDual} or FISTA \cite{FISTA_beck2009fast}.  

In general, unless noted otherwise , the matrix $A$  is calculated once according to  Weickert's shceme, Eq. \eqref{Eq_Wickert_D_tilde} and Eq. \eqref{Eq_Wickert_A},
%, Weickert shceme
%\GG{According to Eq. (ADD), }\SB{done} \GG{ADD here and everywhere in the paper you refer to an equation with eqref and not ref, for example  Eq. \eqref{Eq_Wickert_D_tilde} }\SB{Done}
based on the  input image. To calculate the matrix, we use the values $m=4$  and $c_m = 3.3$
%$c_m = 3.31488$  
for the function $c(x;K)$ , as suggested by Weickert in 
\cite{AnisoWeickert}. 
%\off{The value of $K$  was chosen to be the mean of all $\mu_1$  eigenvalues, i.e. the larger eigenvalue of the two generated per pixel.}\SB{off, ,,,,removed sentence}
The default numerical framework used in this work was Chambolle's \cite{chambolle2004algorithm} due to its high precision and despite its relatively slow convergence rate.
The gradient flow of \eqref{eq_AATVflow} is solved as in \cite{Gilboa_spectv_SIAM_2014}, by iterative proximal steps. 
\off{
The way one does that is by taking the flow equation $u_t = - p$ where p is the sub-gradient of the appropriate one-homogeneous functional $J(u)$ and discretize using Moreau-Yosida implicit approximation \cite{moreau1965proximite} (which is unconditionally stable for $dt$) as follow,
\begin{equation}
    u_t = -p \rightarrow \frac{u_{n+1} - u_n}{dt} = -p(u_{n+1}) \rightarrow   dt p(u_{n+1}) + u_n  - u_{n+1}  = 0.
\end{equation}
We can see it is coinside with the Euler-Lagrange of the ROF problem described in Eq. \eqref{eq_AATVROFModel} where $t=dt$.
}
%\begin{equation}
%    \min_u \mathscr{E}(u) = dt J(u) + %\frac{1}{2}\|f-u\|^2_{L_2}
%\end{equation}
\subsection{Chambolle-Pock Scheme}
In \cite{ChambollePock_PrimalDual} a general scheme is proposed to solve problems of the form,
\begin{equation}
\min_{x\in X} F(Kx) + G(x).
\end{equation}
where $F,F^*:Y \rightarrow [0,\infty)$, $G,G^*:X \rightarrow [0,\infty)$, $K:X \rightarrow Y$. Here, $F^*$ and $G^*$ are the convex conjugate of $F$ and $G$. $K^*$ is the hermitian adjoint of the operator $K$. We assume that these problems have at least a solution $(\hat{x},\hat{y})\in X\times Y$.
The accelerated Arrow-Hurwicz algorithm is used.
%\GG{No need to write equation again..}\SB{Done}
Recalling the ROF-A$^2$TV energy Eq. \eqref{eq_AATVROFModel}, the definitions of $F$, $G$, $K$ and $K^*$ are,
\begin{equation}
\begin{aligned}
F(x) &= \|x\|_{L_1(\Omega)},\\
G(u) & = \frac{1}{2t} \|u-f\|_{L_2(\Omega)}^2,\\
K &= \nabla_A = \Grad_A\\
K^* &= \nabla_A^T = \Div_A.
\end{aligned}
\end{equation}
%\GG{Need to explain how a flow is being implemented using the minimization.}\SB{Done}

\subsection{Chambolle's Projection Algorithm}
The gradient-divergence pair defined in \eqref{eq_GradA} and \eqref{eq_DivA} allows us to modify the projection algorithm of \cite{chambolle2004algorithm} by the following,
\begin{equation}
\label{eq_Chambolle_projection}
 p_{i,j}^{n+1} = \frac{p_{i,j}^{n}+\tau \left( \nabla\left( \Div_A(p^n)-2\lambda f \right) \right)_{i,j}}{1+\tau \left| \left( \nabla\left( \Div_A(p^n)-2\lambda f \right) \right)_{i,j} \right|},
\end{equation} 
where the dual parameter is $p\in\partial J_{A^2TV}(u)$, $f$ is the original image, and $\tau$ is the time step. According to \cite{chambolle2004algorithm}, convergence is guaranteed for $\tau<1/8$ with the outcome of $u = f - \frac{1}{\lambda}\Div_A(p)$. %\SB{added u=...}
%(but heuristically happens for $\tau<1/4$ as well).
It is worth noting that according to \cite{chambolle2004algorithm}, $\tau$ is limited by $\frac{1}{\|div_A\|^2_{L_2}}$. Since the eigenvalues of the matrix $A$ in our model are in the range $(0,1]$, the norm of the divergence can only diminish
%\GG{grow or decrease? If the norm grows the bound gets smaller..}\SB{Fixed it..}
; thus, one can maintain the same bound on $\tau$. 
%\GG{SHOULD CHECK THIS. The parameter $\tau$ relies on the norm of the operator. We need to see if this changes here.}\SB{Done.}

\section{Numerical Experiments}
\label{sec:exp}

%%%% BEGIN - FIGURES OF EXPERIMENTS, SECTION 6

\begin{figure}[htb]
\centering
\includegraphics[width = 0.7\textwidth]{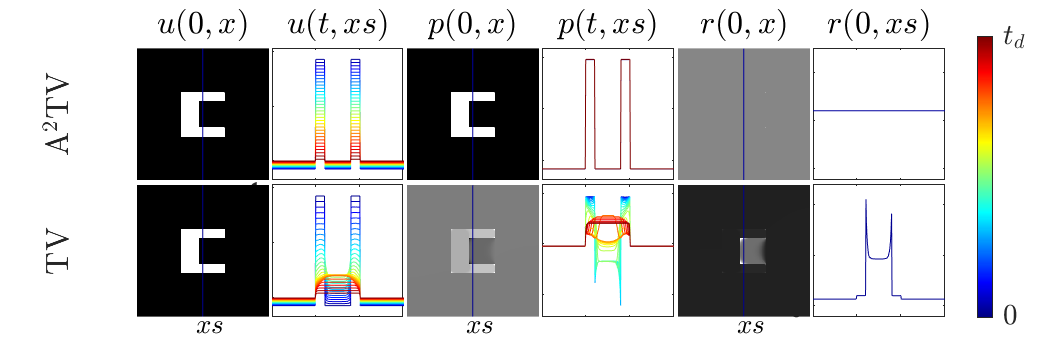}
\caption{Non-Convex EigenFunction. Top - A$^2$TV($a=0.25$) flow of a $C$ shape. Bottom - TV flow of a $C$ shape. The $C$ shape has a convex hull perimeter to shape perimeter ratio of 0.769. From left to right: the image $u$ at $t=0$, a vertical cross section of the image $u$ over time (blue is the start and red is the end of the flow), the sub-gradient $p$ at time 0, a cross section of the sub-gradient $p$ over time, the ratio image as described in Eq. \eqref{eq_ratio_p_over_u} and a vertical cross section of the ratio image.}
\label{fig:Non-Convex_1_EigenFunction}
\end{figure} 

\begin{figure}[htb]
\centering
\vspace{-10px}
\includegraphics[width = 0.65\textwidth]{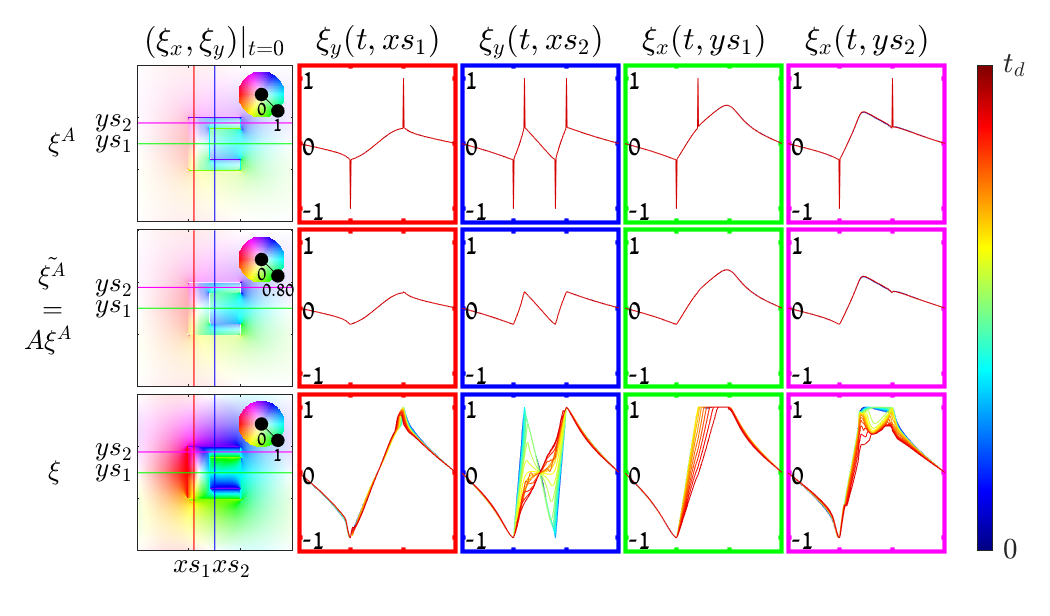}
\caption{Vector field $\xi$. Top and Middle, A$^2$TV flow of a $C$ shape. Bottom, TV flow of a $C$ shape. On the left are the $\xi$'s of the flows at $t=0$. The colored axes are cross-sections of $\xi$ at different evolution times $t$, where the time is color-coded by the bar on the right (for A$^2$TV the plots coincide and are occluded by the last instance in red) . Middle row, the applied tensor $A$ on $\xi^A$, Eq. \eqref{eq_tilde_xi^A}.}
\label{fig:Non-Convex_2_Xi}
\end{figure} 

\begin{figure}[htb]
\centering
\includegraphics[width = 0.7\textwidth]{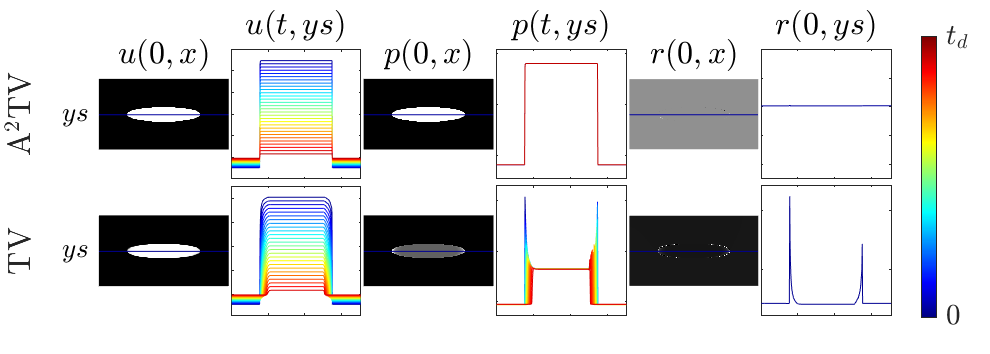}
\caption{Extreme-curvature convex eigenfunction. Top - A$^2$TV flow. Bottom - TV flow. For both flows, the initial condition is the same ellipse shape. From left to right: the image $u$ at $t=0$, a horizontal cross section of the image $u$ over time (blue is the beginning and red is the end of the flow), the sub-gradient $p$ at $t=0$, a cross section of $p$ over time, the ratio image $r$ at $t=0$, Eq. \eqref{eq_ratio_p_over_u}, and a vertical cross section of it.}
\label{fig:Extreme-Curvature_1_EigenFunction}
\end{figure} 
\sidecaptionvpos{figure}{c}
\begin{SCfigure}[1.1][htb]
\centering
\includegraphics[width = 0.45\textwidth]{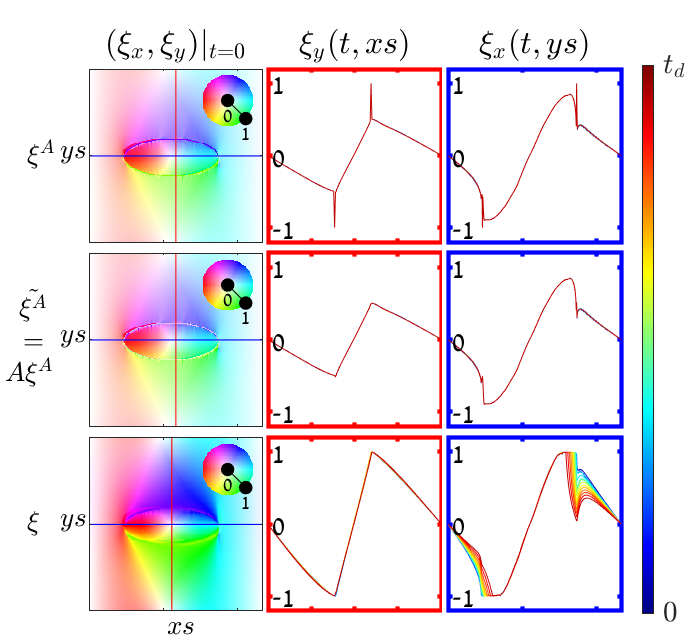}
\caption{Extreme-Curvature $\xi$. Top and Middle - A$^2$TV flow of an ellipse shape. Bottom - TV flow of an ellipse shape. On the left are the $\xi$'s of the flows at $t=0$. The colored axes are the variable $\xi$ cross-sections evolves through time $t$. The middle row is the applied tensor $A$ on the A$^2$TV flow $\xi^A$.}
\label{fig:Extreme-Curvature_2_Xi} 
\end{SCfigure}

\begin{figure}[htbp]
\centering
\includegraphics[width = 0.95\textwidth]{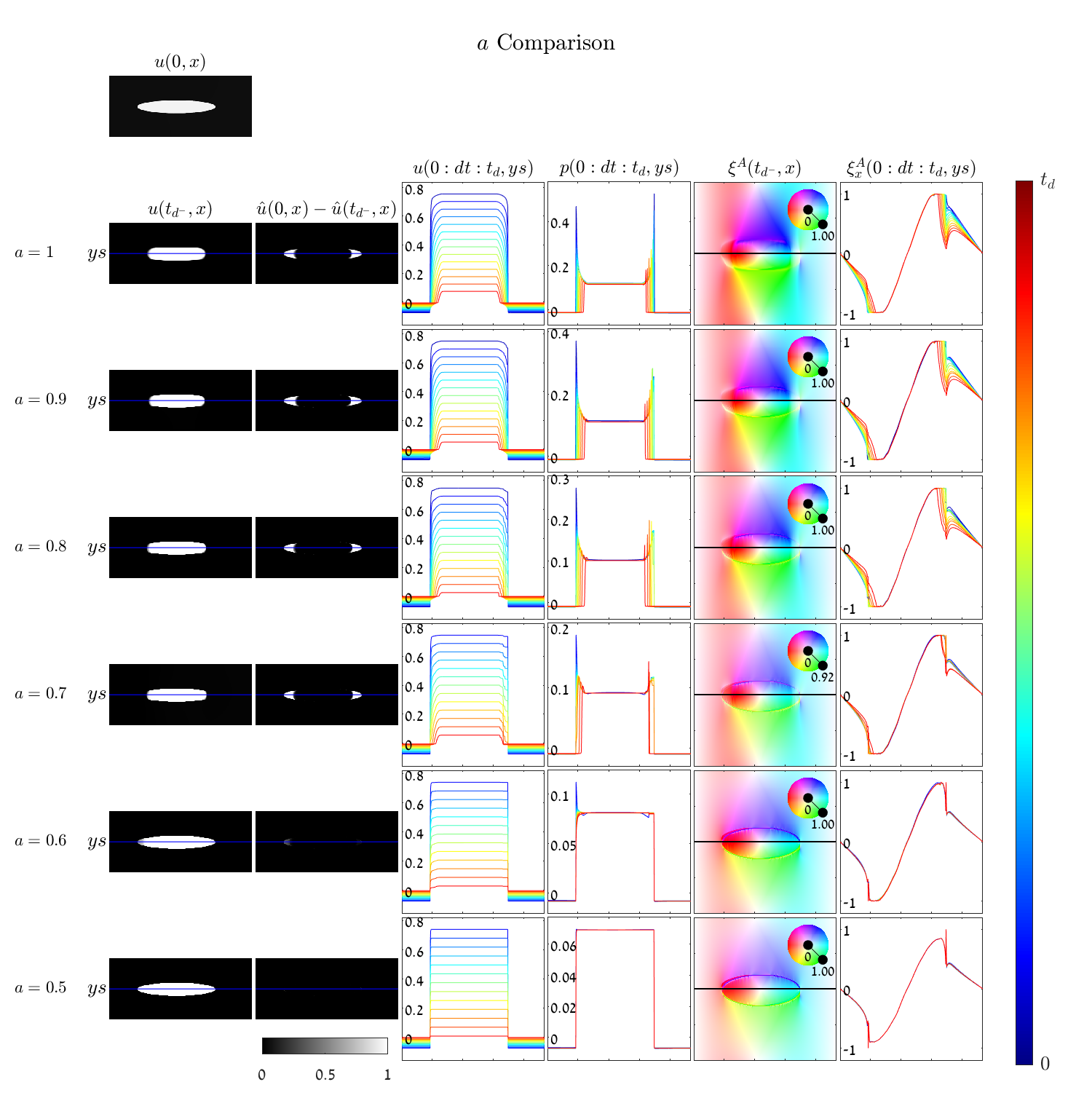}
\caption{Extreme-curvature for different values $a$. Top, ellipse image $u(t=0)$. Each row depicts an A$^2$TV flow with its corresponding $a$. From left to right, the image $u$ just before vanishing time, 
the difference image between start and vanishing times (normalized),
%the one at the start of the flow and the one at its end after the normalization of each ($\hat{u}$),
horizontal cross-section of the image $u$ over time $t$ along the flow, horizontal cross-section of the sub-gradient $p$, $\xi^A$ before vanishing time and a cross-section of $\xi^A$ along the flow.}
\label{fig:Extreme-Curvature_4_a_Comparison}
\end{figure} 

\begin{figure}[htb]
\centering
\includegraphics[width = \textwidth]{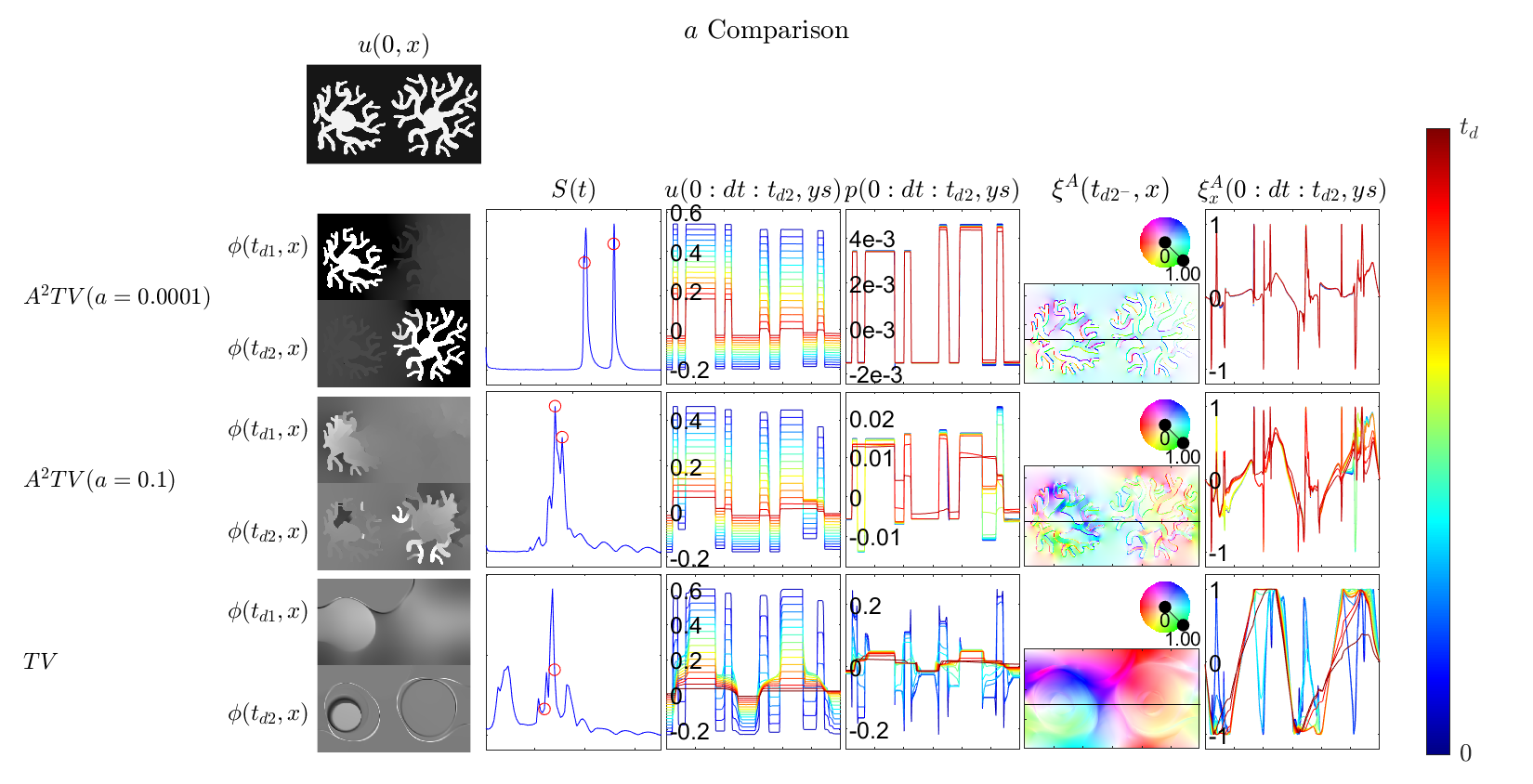}
\caption{Synthetic neuron-like structure - comparison for different values of $a$. Top is a pair of neuron-like shapes image $u(t=0)$. Each row is an A$^2$TV flow with its corresponding $a$. 
%\GG{Spectral stuff, too complicated new info to be put only in the caption. You should explain better what you did, how you seperated the figures, do you have the spectrum to show some peaks?}
From left to right, the $\phi$ spectral component (narrow band-pass) just before the neuron shape vanishes, $t_{td1}$ for the right neuron and $t_{td2}$ for the left one, horizontal cross-section of the image $u$ over time $t$ along the flow, horizontal cross-section of the sub-gradient $p$ over time $t$ along the flow, $\xi^A$ just before the first neuron shape vanishes and a cross-section of $\xi^A$ over time $t$ along the flow. }
\label{fig:Synthetic_Example_5_Neurons_a_Comparison}
\end{figure}

 %%%% END - FIGURES OF EXPERIMENTS, SECTION 6

In this section we present numerical experiments that illustrate and further validate theoretical issues discussed so far.
Let us define the ratio image
\begin{equation}
    r(t,x) := \frac{p(t,x)}{u(t,x)}.
    \label{eq_ratio_p_over_u}
\end{equation}
For eigenfunctions, we get $r = const = \lambda$.
Throughout this section we present the flow in terms of the image $u$, the sub-gradient $p$, the ratio $r$ and the dual variable vector field parts $\xi_x$ and $\xi_y$.

Plots of cross-sections at various time instances are super-imposed  with an appropriate color-bar, with $t \in [0,t_d]$, where $t_d$ denotes time of disappearance (extinction time). When a two-dimensional image of the vector field $\xi$ is shown, a color code is used to show its size and angle for each pixel (the code is shown by a circle on the top-right). To achieve high precision, 3000 projection iterations are used at each step.
%In these experiments Chambolle's projection algorithm was used due to its ease of use and high precision. %The way it was used is by initializing the dual variable with the one from its previous run; and using 
%(each projection was computed by ).

\paragraph{Non-Convex Shape Experiments}
In order to better illustrate the difference between the A$^2$TV and TV flows, we applied both flows on an indicator function of a set which resembles a $C$ shape, with a ratio  $Per(hull(C))/Per(C)$ of $0.769$ (Eq. \eqref{eq_AATV_a_max__NonConvexity}) as can be seen in Fig. \ref{fig:Non-Convex_1_EigenFunction}. 
%
%\GG{Relate it to the theory, what is $a$? what is the ratio as given in equation \eqref{eq_AATV_a_max__NonConvexity}  for this shape?}\SB{mostly in the figure caption}
%
We recall the analytic solution of eigenfunctions under TV-flow, Eq. \eqref{eq_TVflowSolution}, and A$^2$TV-flow,  Eq. \eqref{eq_AATVflowSolution}, in which spatial structure is preserved and contrast is reduced linearly with time.
It can be clearly seen that
 TV-flow deforms the $C$ shape over time, indicating it is not an eigenfunction, as classical TV theory predicts for non-convex sets. On the other hand, A$^2$TV-flow keeps a steady sub-gradient, preserves the shape and  yields a constant ratio function $r$. 
%\GG{Explain this image here, not in caption.}\SB{In what way? the image's caption only explain what is inside the figure, not why or the meaning of any of it..}
This is a strong numerical indication that the shape is an eigenfunction of A$^2$TV, in the sense of \eqref{eq_ef_AATV}. 
The flow can also be investigated by examining the dual domain vector field  $\xi$.
%\GG{Be consistent here with $\xi$ and the one shown in Fig. \ref{fig:Curvature_Conjecture_6__2_AATV_Ellipses_Flows}. Looks like there $\xi = A \xi^A$.}
In Fig. \ref{fig:Non-Convex_2_Xi} the variables $\xi^A$,$\tilde{\xi^A}$ are shown. Different instances in time are overlaid. One can observe they coincide with almost no change throughout the A$^2$TV flow, while $\xi$ is being distorted in the TV flow.
The row in the middle shows $\tilde{\xi^A}$, which is the variable the divergence is applied on, similar to $\xi$ in TV. 

%\clearpage

\paragraph{High Curvature Experiments}
In this section we want to demonstrate how extremely curved objects can still be eigenfunctions of A$^2$TV and compare these to the TV case.
%
%\GG{I don't think we show the part below, consider to remove }\SB{fixed. fixed the disappearing figure bug}
%
%\GG{======================}
%
In Fig. \ref{fig:Extreme-Curvature_1_EigenFunction}, it can be seen that the ellipse linearly diminishes over time in the A$^2$TV flow while reducing the curvature of the ellipse in the TV flow. The constant ratio image (up to numerical errors) and the constant sub-gradient strongly validate that this highly curved ellipse is an eigenfunction of the A$^2$TV functional.

In Fig. \ref{fig:Extreme-Curvature_2_Xi} the  dual variable is shown.
Note that the magnitude of the variable $\tilde{\xi^A}$, Eq. \eqref{eq_tilde_xi^A}, is smaller than 1 on the set boundary and may be close to 1 inside the set, compensating for high curvatures, as seen along the $x$ axis (blue frame in middle row, right).
%$\xi^A$, its $A$ multiplied one $\tilde{\xi^A}$ as described in Eq. \eqref{eq_tilde_xi^A} and the TV dual variable $\xi$. 
%Similar to the non-convex shape, one can see that the 
%\GG{again here, confirm notations of $\xi$ variants with previous sections, define by an equation.}
%$\tilde{\xi^A}$ generalizes the TV's one in a way that it allows 3rd degree polynomial like shapes inside the ellipse as oppose to the linear case in the TV case.

%\GG{====================================}

In Fig. \ref{fig:Extreme-Curvature_4_a_Comparison} we demonstrate the impact of the parameter $a$ on the flow variable $u$, the sub-gradient $p$ and the dual variable $\xi^A$.  Results of several A$^2$TV flow processes, Eq. \eqref{eq_AATVflow}, with different parameter $a$  are shown. As $a$ descreases the ellipse structure becomes less regularized, where for $a=0.5$  the ellipse becomes a numerical eigenfunction and its shape is completely retained through the flow.
\off{
\GG{Put here a different explanation.}
In Fig. \ref{fig:Extreme-Curvature_4_a_Comparison}, we demonstrate how the change of the anisotropy parameter $a$  affects $p$, $\xi^A$ and the flow $u(t,x)$.
It can be seen that as $a$ decreases, the strong regularization of high curvature regions diminishes, until for some value of $a$ (here $a=0.5$, bottom row) we get that $u$ is an eigenfunction ($p=\lambda u$), and the flow decreases linearly, as can be predicted from \GG{ADD eq. of flow for eigenfunctions from theoretical section.} 
%functions as TV, and as $a$ decreases it functions as an A$^2$TV which has a wider set of eigenfunctions.
}

%\pagebreak
\paragraph{Synthetic Example - Neuron-like Structure}
%\SB{Soo.. because we have the shapes example, Im thinking maybe remove this one...}
%\GG{I think this example is nice, still a toy but much more complex, couples well both non-convexity and high curvature areas. The difference from TV here is very clear (as long as you explain well what you do, basically you can say you show results of different band-pass filters (just define it in the spectral section).}
The purpose of this example is to examine a complex, non-convex and highly curved shape. We also demonstrate the result of A$^2$TV spectral filtering on such shapes. In this example, the use of BPF's of the spectral framework,  Eq. \eqref{eq_BPF}, were used at scales $t$ near the vanishing time of each neuron. As $a$ diminishes  the spectrum $S(t)$ is more sparse, retaining most of the structures within distinct scales with compact support in $t$. For small enough $a$, one can see a linear decay of the image $u$ over time. Similar to previous experiments,
 $p$ and $\xi^A$ do not change over time, as can be expected from an eigenfunction.
%along the flow which numerically proves that there exist a critical $a$ that beneath it, both the neurons are eigenfunctions.
For different values of $a$, note the different regularization that occurs, in terms of both  convexity and curvature.  
%the same time it is quite interesting to see the A$^2$TV flow changes from the TV where $a=1$ to the A$^2$TV flow where the neurons are its eigenfunction.
%\GG{All the spectral stuff you put in the caption should be put here with more details on the filters you used.}
\clearpage

\section{Applications}

\begin{figure}[htb]
\centering
\includegraphics[width = 0.9\textwidth]{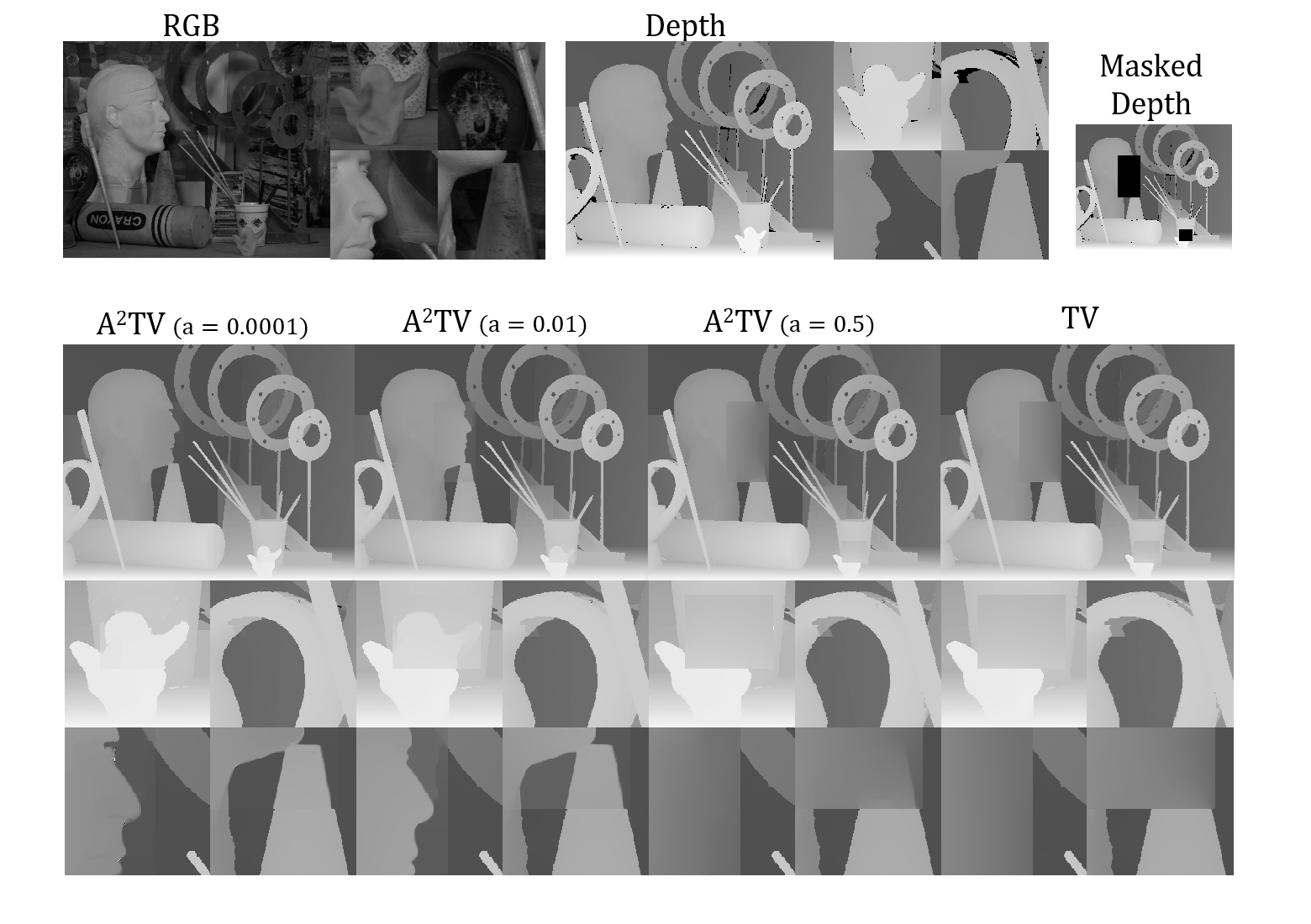}
\caption{Image-guided depth inpainting - $a$ Comparison. Top row, left to right: RGB image, original depth image and the masked depth image. Second row: Inpainting with the A$^2$TV functional with various $a$ parameter values and one with the TV functional. }
\label{fig:APP_Depth_Inpainting}
\end{figure} 

\begin{figure}[htb]
\centering
\includegraphics[width = \textwidth]{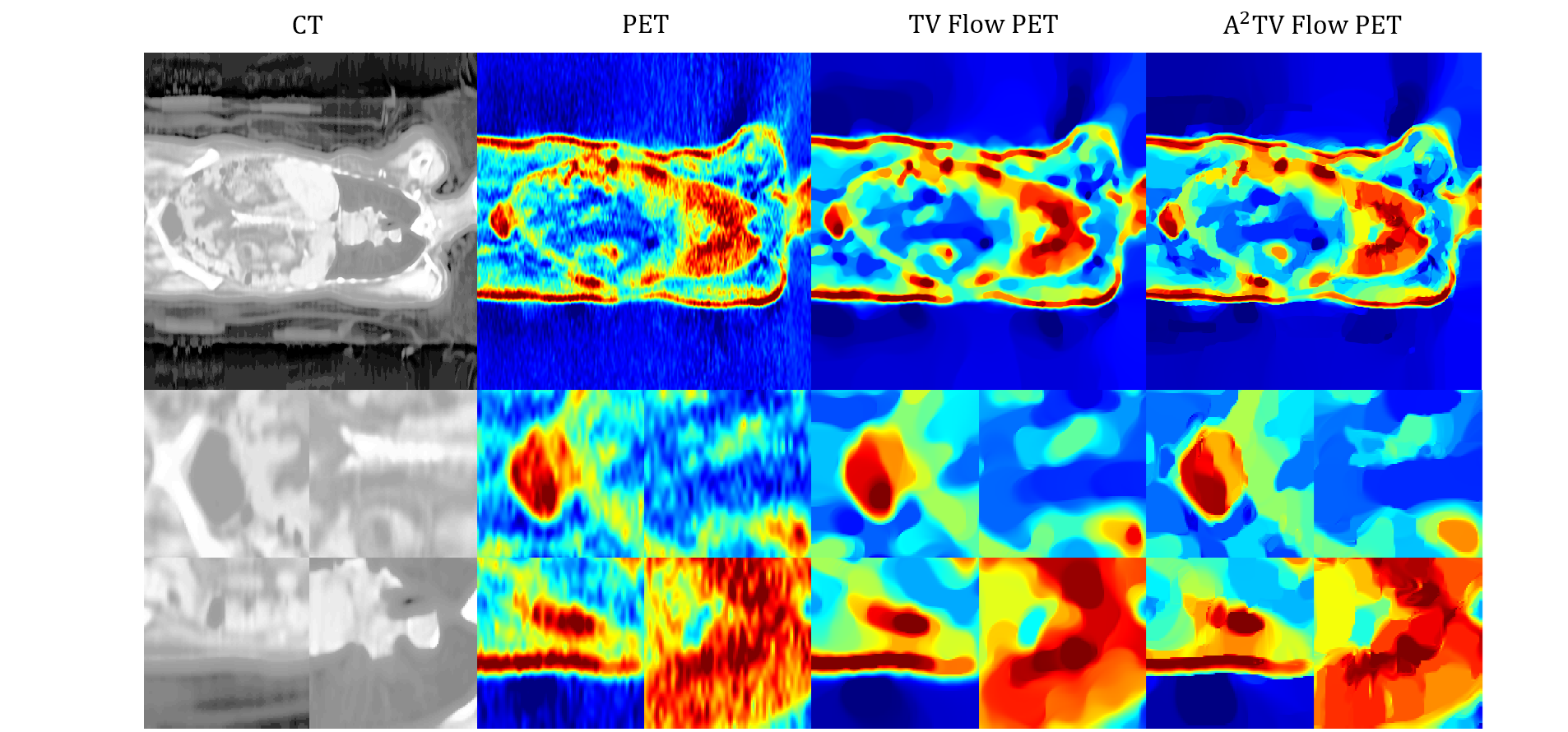}
\caption{PET/CT Model Fusion. High resolution information on the anatomy of the patient is injected to the PET regularization through the tensor $A$, which is based on the registered CT image. }
\label{fig:APP_PETCT_Fusion}
\end{figure} 

\begin{figure}[htb]
\centering
\includegraphics[width = 1\textwidth]{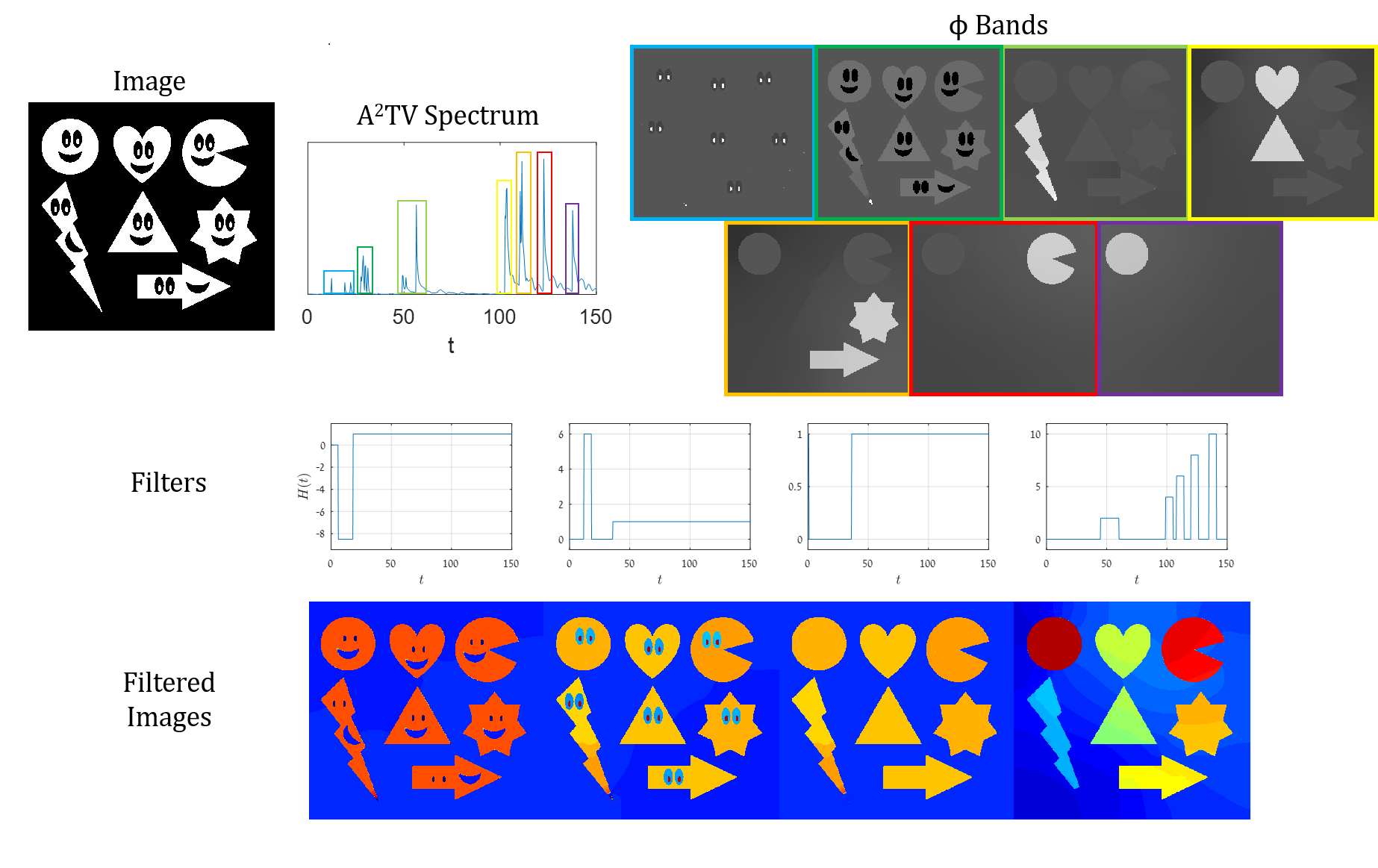}
\caption{Spectral A$^2$TV - Shapes and Faces. Top left to right - An image of objects with faces; its $A^2TV$ spectrum and $\phi$ bands of a specific locations in the spectrum. Middle and bottom - Filters applied on the spectrum and the resulting filtered image beneath each one.
}
\label{fig:APP_SAATV_faces}
\end{figure} 

\begin{figure}[htb]
\centering
\includegraphics[width = 1\textwidth]{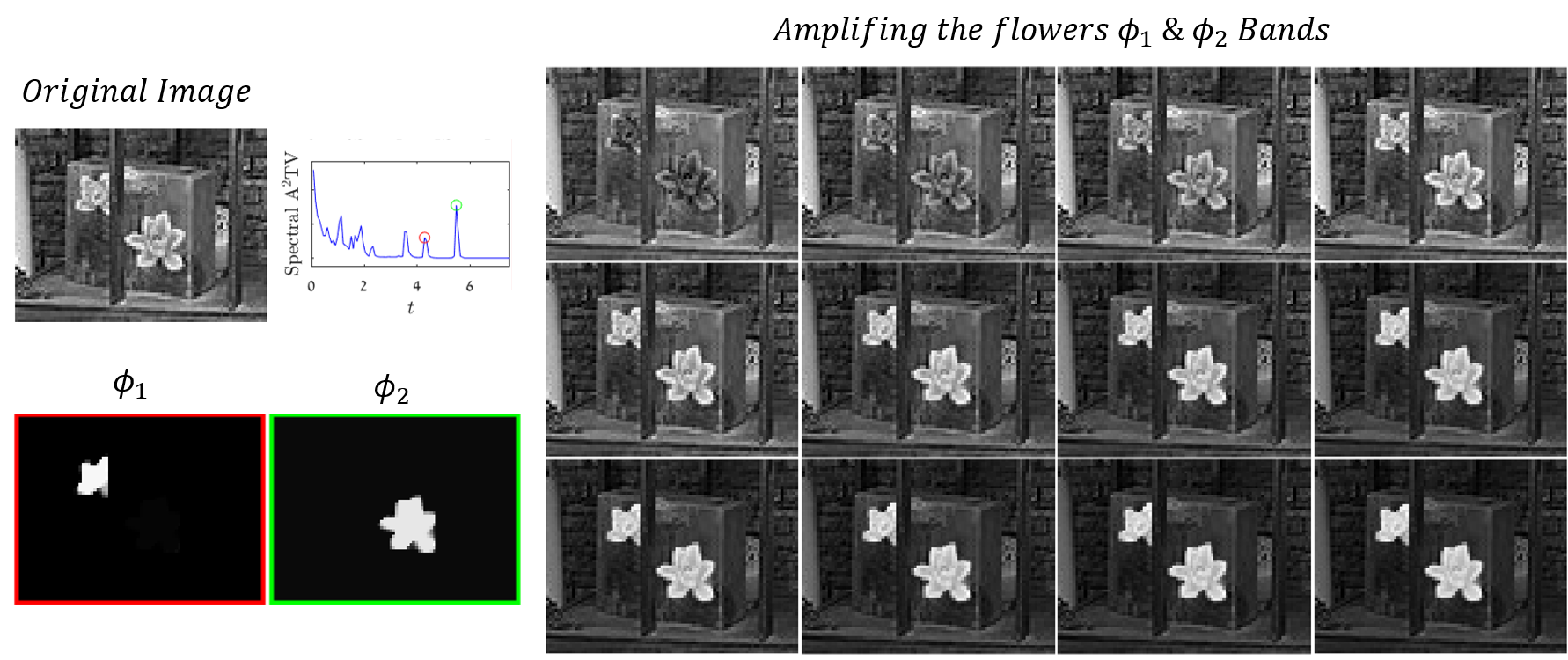}
\caption{Spectral A$^2$TV - Attenuating and enhancing object bands. Spectral peaks in the spectrum (red, green circles) correspond to the small and larger flowers, respectively. These bands are filtered out completely at the top left image. The bands are then added back by  33\% increments  (from left to right, top to bottom).
}
\label{fig:APP_FlowersBox}
\end{figure} 

We show below several applications which can benefit from the A$^2$TV functional, following our analysis, in a minimization and spectral frameworks. 
\subsection{Depth Inpainting}
%Fast image inpainting and colorization by Chambolle’s dual method
In this application, we construct the adapted matrix (Tensor) $A$ from the RGB image,and apply it to the depth image. The energy minimization is based on \cite{App_Inpainting_li2011fast}, where now A$^2$TV replaces TV. The Energy being minimized is, 
\begin{equation}
\label{eq_InpaintingEnergy} 
\min_{u\in BV(\Omega)}\mathscr{E}(u) = J_{A^2TV}(u) + \frac{1}{2}\int_\Omega\hat{\mu}(x)\left(u(x) - f(x)\right)^2dx,\hspace{3mm} A(x) = A(I(x)),
\end{equation}
%\GG{Changed here some of the explanations. Check that it's OK. Add a few words on why do we need to iterate (this seems like a convex problem), is it faster? Also how you initialize v, do you change values of $\mu$? How many iterations between the sub problems are used? Do you wait for convergence of Chambolle's projection or you do only a single iteration?}\SB{Done.}
%
where $u$ is the inpainted depth image, $f$ is the input depth (with missing values), $I$ is the intensity of the RGB image and $\hat{\mu}$ is an inverse indicator of the inpainting regions with zero value in regions to be inpainted and $\mu \gg 1$ otherwise.
Standard projection algorithms cannott solve Eq. \eqref{eq_InpaintingEnergy} directly since $\hat{\mu}$ takes zero values  in $\Omega$ (
recalling $u = f - \frac{1}{\mu}\Div_A(p)$ in the projection algorithm).
This is resolved in \cite{App_Inpainting_li2011fast} by solving two subproblems alternately and iteratively,
\begin{equation}
\label{eq_InpaintingSubproblem1} 
\min_v J_{A^2TV}(v) + \frac{1}{2\theta}\int_\Omega\left(u - v\right)^2dx,
\end{equation}
\begin{equation}
\label{eq_InpaintingSubproblem2} 
\min_u \int_\Omega\left(u - v\right)^2dx + \frac{1}{2}\int_\Omega\hat{\mu}\left(u - f\right)^2dx.
\end{equation}
Problem \eqref{eq_InpaintingSubproblem1} can be solved by the adaptation of Chambolle's projection method \cite{chambolle2004algorithm}, Eq. \eqref{eq_Chambolle_projection}.
%\begin{equation}
%\label{eq_InpaintingSubproblem1_sol} 
%p^{n+1} = \frac{p^n + \tau \nabla_A \left(\Div_A p^n -u/\theta\right)}{1 + %\tau |\nabla_A \left(\Div_A p^n -u/\theta\right)|}; \hspace{5mm}  v = u-\theta\Div p^*.
%\end{equation}
Problem \eqref{eq_InpaintingSubproblem2} has the following closed-form solution,
\begin{equation}
\label{eq_InpaintingSubproblem2_sol} 
u = \frac{\hat{\mu}\theta f + v}{1 + \hat{\mu}\theta}.
\end{equation}
In our scheme, the adapted anisotropic matrix $A$ is calculated once at the beginning from the RGB image, $p$ is initialized at first with zeros and later with the result of previous iteration, $v$ is therefore derived from $p$ by $v = u - \theta\Div_A(p)$. Moreover, $u$ is initialized as $f$, and the parameters are $\mu = 80, \theta = 5, \tau = 1/8$ with 6000 iterations. At each alternating session, the projection step is iterated 5 times. 
The resulting image in Fig. \ref{fig:APP_Depth_Inpainting} with the smallest $a$ clearly shows how the A$^2$TV regularizer is able to reconstruct the required geometry given a small enough parameter $a$.

\subsection{PET/CT Model Fusion}
Here, we construct the $A$ tensor from the CT image, and solve a TV and an A$^2$TV minimization problems. In the A$^2$TV minimization process,  information regarding the structure of the PET scan is taken from the detailed CT image via the tensor $A$. The energy to be minimized is,
\begin{equation}
\label{eq_PETCT} 
\min_{u\in BV(\Omega)}\mathscr{E}(u) = J_{A^2TV}(u) + \frac{\mu}{2}\int_\Omega\left(u - f\right)^2dx,\hspace{3mm} A(x) = A(CT(x)),
\end{equation}
where $f$ is the PET image, $\mu = 5/3$ and the adapted matrix parameter $k = 0.1$. Results are shown in Fig. \ref{fig:APP_PETCT_Fusion}.

\subsection{Spectral A$^2$TV - Shapes and Faces}
As an application to the spectral A$^2$TV, we chose the "Shapes and Faces" cartoon image in order to better demonstrate how non-convex and high curvature shapes are decomposed and can be manipulated using the spectral A$^2$TV framework.
Various band-pass filters, Eq. \eqref{eq_BPF}, are used  to isolate shapes and certain face parts of different scales.
%The original image in Fig. \ref{fig:APP_SAATV_faces} has been decomposed to its basic components through the process of A$^2$TV flow as described in Eq. \eqref{eq_Spectral_Framework_Flow} only to compose the transform described in Eq. \eqref{eq_Spectral_Framework_Phi_Transform}. Which can be seen in the figure as $\phi$ bands and the spectrum described in Eq. \eqref{eq_Spectral_Framework_Spectrum} as A$^2$TV spectrum. A filtration process has been done by applying different filters using Eq. \eqref{eq_Spectral_Framework_Filtering} and returning to the image domain using Eq. \eqref{eq_Spectral_Framework_inverse_Phi_Transform}. 
As can be seeing in the figure, by applying different filters we are able to extract different parts of the image and apply various filters (middle row) yielding different semantic interpretations (bottom row).
\subsection{Spectral A$^2$TV - Attenuating and enhancing "low frequency" object bands}
We now show manipulation of a natural image based on A$^2$TV spectral representation. Sharp peaks in the spectrum $S(t)$ at high $t$ (corresponding to low eigenvalues or "low frequencies") are isolated by narrow band-pass filters. See spectrum at top row, second plot and shapes corresponding to the first and second peaks at the bottom row, left.
One basically obtains a segmentation of the flowers. 
%
%In this application we changed the two flower's  $\phi$ bands within a real life image. Here, due to the generalized nature of the A$^2$TV eigenfunction, the flowers shape hasn't been distorted, 
%
We can now perform filtering on the appropriate bands, increasing or decreasing the contrast of the flowers. 
%Here, similar to the previous application, a spectral transform was used by performing A$^2$TV flow as described in Eq. \eqref{eq_Spectral_Framework_Flow} and a transform which is described in Eq. \eqref{eq_Spectral_Framework_Phi_Transform}.
%A filtering was used as described in Eq. \eqref{eq_Spectral_Framework_Filtering} and returning to the image domain using Eq. \eqref{eq_Spectral_Framework_inverse_Phi_Transform}. 
We note that one is not able to process objects in a similar manner using spectral TV as it would have distorted the flowers' segmentation to a disk-like shapes and  filtering would not be accurate.

%%%%%%%%%%%%%%%%% Removed (GG) %%%%%%%%%%%%%%%
\off{
\subsection{Optoacoustic Spectral A$^2$TV based Reconstruction}
Here, we modified the SA$^2$TV support image reconstruction from opto-acoustic signals.
This example is taken from a joint work with colleges that will be published soon.
\begin{figure}[H]
\centering
\includegraphics[width = 0.9\textwidth]{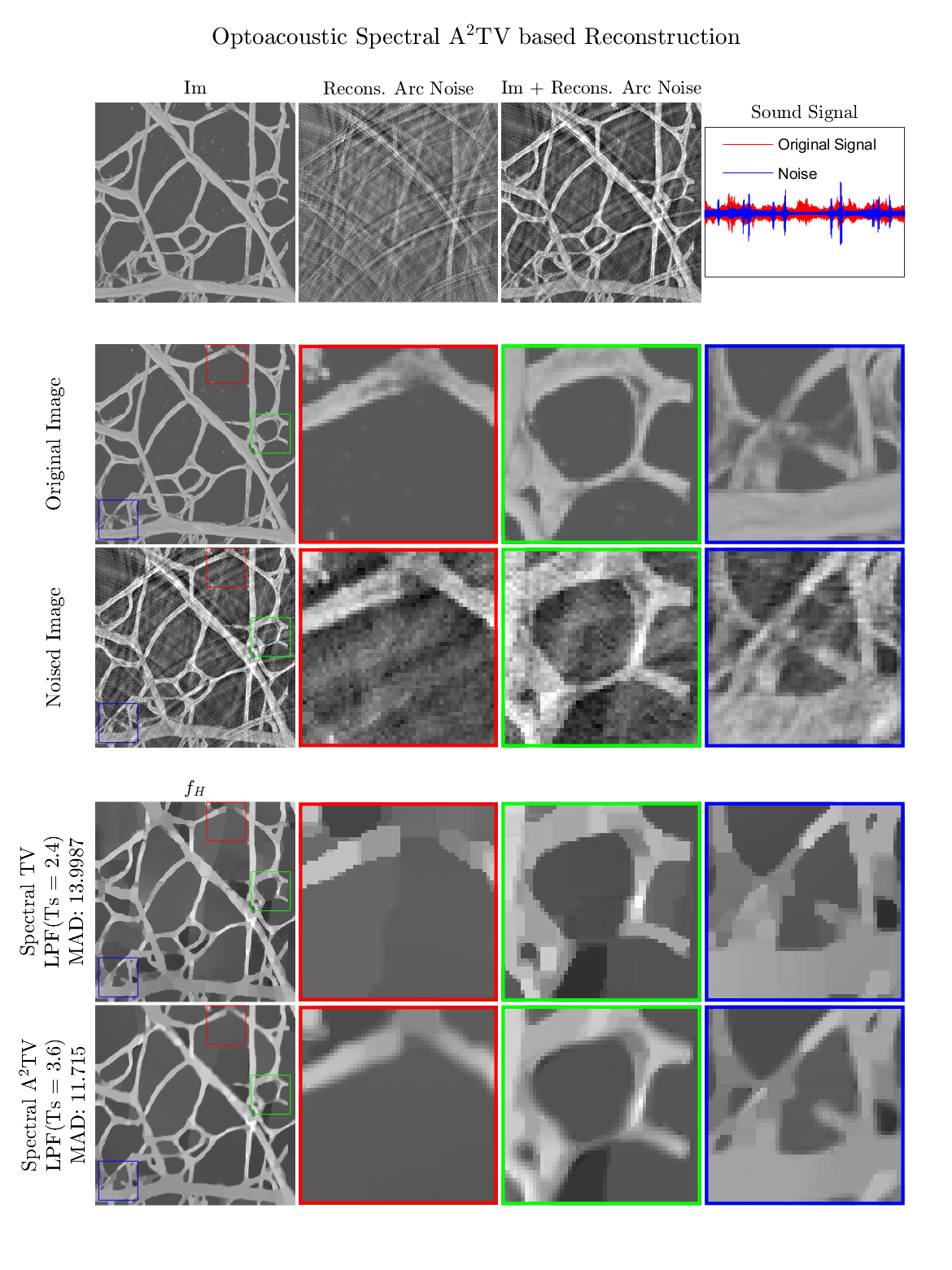}
\caption{Optoacoustic Spectral A$^2$TV based Reconstruction.
\GG{For this application a lot of effort was dedicated. I expected results to be a little better. 
For some reason the AATV is quite blurry. Unlike the synthetic example you showed before, it does not keep the details that well. I don't know why, look like the matrix $A$ was not computed that well.
Probably not easy to correct now, are there other results? At this stage not sure this is worth a full page figure.
}
}
\label{fig:APP_Optoacoustic}
\end{figure} 
}
%%%%%%%%%%%%%%% Removed this part, will appear in the 2nd submission

\section{Conclusions}
In this work, the A$^2$TV functional was analyzed and characterized, focusing on the type of shapes it can perfectly preserve. These shapes are nonlinear eigenfunctions in the sense of \eqref{eq_ef_AATV}. It is shown that
A$^2$TV can be formulated as a smooth generalization of the TV functional, by changing the anisotropy parameter $a \in (0,1]$. 
Its eigenfunctions are a superset of those of TV, where as $a$ tends to zero more complex structures are admissible. We have proved that the convexity measure \eqref{eq:conv_C} can be used to characterize the most nonconvex structure that is an admissible eigenfunction for a certain $a$. With respect to curvature, we have shown numerically the tendency of higher curvatures to be admitted, approximately inversely proportional to $a^3$, compared to isotropic TV. Future theoretical research is required to formally validate this.

Those findings provide better insight on the functional and how to tune its parameters. Since this is a one-homogeneous functional, we were able to generalize spectral TV \cite{Gilboa_spectv_SIAM_2014} to spectral A$^2$TV, following the theory of \cite{Gilboa2016}. This type of nonlinear decomposition is well suited not only for filtering the textural parts of the image, but the objects as well. Finally, several applications that may benefit the use of A$^2$TV were suggested.

\off{
from 1, which reduces the A$^2$TV to TV, to almost zero), and a much more robust in the sense of maintaining objects within an image without a tendency to change the nature of an object regardless of its convexiness and its curvature.
Our research provided a way to evaluate the $a$ parameter or at-least better understand its implications on the degree of the convexiness and curvature of objects within an image.
A dual formulation of the A$^2$TV has been also introduced, which helped adjusting the existent TV based algorithms in our work.
Moreover, the spectral A$^2$TV has been presented and its attributes has been shown. Thereby enabling the full use of the spectrum to not only accommodate the texture information of the image, but also the object one as well. 
In order to show the massive advantages of this functional and the knowledge that the theory brings to the table we handpicked several applications that we thought brings the best out of the A$^2$TV functional.

It goes without say that the curvature property needs to be analytically proven, probably by using tools outside of the convex sets field work.  
}

%\section{Appendices}
\section{Appendix I - A$^2$TV Disk Example}
We define the following function,
\begin{equation}
f(X) = \Bigg\{
\begin{array}{lc}
hc_0  & |X|\leq R\\
h(c_0-1) & R\leq |X|\leq R_0
\end{array}.
\end{equation}
The shape of $f$ is basically a disk of radius $R$ and height $h$, within a circular domain $\Omega$ of radius $R_0$. Here, $c_0=1-R^2/R_0^2$ was chosen so the mean value is $\bar{f}=0$.
We define the following matrix $A$ as in \eqref{eq_Dtilda},
\off{
\begin{equation}
\label{eq_AVD}
A = V^TDV,\hspace{3mm} V = [v_1;v_2],\hspace{3mm} 
D = 
\begin{cases}
\begin{pmatrix}
a &0 \\ 0 &1
\end{pmatrix}\textrm{, on} \hspace{2mm} \partial\Omega\\
I \hspace{13mm} \textrm{, otherwise}
\end{cases}.
\end{equation}
}
with $\partial C: |X|=R$. Eq. \eqref{def_V_D} can be written as,
\begin{equation*}
v_1 = \frac{X}{|X|} = \frac{X}{R} = \frac{1}{R} \svec{x}{y}, \hspace{4mm}
v_1 \perp v_2 \Rightarrow v_2 = \frac{1}{R} \svec{y}{-x} , \hspace{4mm}
V = [v_1;v_2] = \frac{1}{R} \smat{x}{y}{y}{-x},
\end{equation*}
\begin{equation*}
A = V^T \Lambda V = \frac{1}{R}\smat{x}{y}{y}{-x}\smat{a}{0}{0}{1}\frac{1}{R}\smat{x}{y}{y}{-x} = \frac{1}{R^2}\smat{ax^2 + y^2}{axy-xy}{axy-xy}{ay^2+x^2} .
\end{equation*}
In this eigen-problem the eigenvalue of $v_1$ is $a$ and of $v_2$ is $1$.
% We know that on the edge, the adjoint vector $\xi^A$ is parallel to the gradient of the image on the shape, pointing outwards and is equal to one. and that $p$ is the eigenfunction and equal to, 
Let us assume a certain vector field $\xi^A$ and show it is a calibrable set. On the disk boundary $\partial C$ we require,
\begin{equation*}
\xi^A|_{\partial C} = \frac{1}{R}\svec{x}{y} = \frac{X}{R} .
\end{equation*}
The divergence operator is applied on  $A\xi^A$, therefore we impose continuity to get,
\begin{equation}
\xi^A = 
\begin{cases}
a\frac{X}{R} &, |X|<R,\\
\frac{X}{R} &, |X|=R,\\
a\frac{R}{c_0}\left(\frac{1}{|X|^2}-\frac{1}{R_0^2}\right)X &, R_0>|X|>R.\\
\end{cases}
\end{equation}
We compute  $p = \Div_A \xi^A = \Div (A\xi^A).$ Thus we obtain,
%\begin{equation*}
%p = \Div_A\xi^A = \Div(A\xi^A) = 
%\begin{cases}
%[A = I] \Rightarrow \Div\left(a\frac{X}{R}\right) \Rightarrow \frac{2a}{R}  %&, |X|<R\\
%[\text{$a$ eigenvector of $v_1$}] \Rightarrow \Div\left(a\frac{X}{R}\right) %\Rightarrow \frac{2a}{R} &, |X|=R\\
%[A = I] \Rightarrow \Rightarrow %\Div\left(a\frac{R}{c_0}\left(\frac{1}{|X|^2}-\frac{1}{R_0^2}\right)X\right) %\Rightarrow -a\frac{2R}{c_0R_0^2}   &, R_0>|X|>R\\
%\end{cases},
%\end{equation*}
%which gives us,
\begin{equation}
p = 
\begin{cases}
\frac{2a}{R}  &, |X| \le R,\\
-a\frac{2R}{c_0R_0^2}   &, R_0>|X|>R.\\
\end{cases}
\end{equation}
One can validate the solution admits $\|\xi^A\|_\infty \le 1$:
\begin{equation}
\|\xi^A\|_\infty = 
\begin{cases}
\norm{a\frac{X}{R}}_\infty < a \le 1 &, |X|<R\\
1  &, |X|=R\\
\norm{a\frac{R}{c_0}\left(\frac{1}{r^2}-\frac{1}{R_0^2}\right)X}_\infty 
\underbrace{<}_* a \le 1 &, R_0>r=|X|>R\\
\end{cases} \Rightarrow \|\xi^A\|_\infty \le 1 .
\end{equation}
\[
\begin{aligned}
* &= 
\norm{a\frac{R}{c_0}\left(\frac{1}{r^2}-\frac{1}{R_0^2}\right)X}_\infty = 
%\norm{a\frac{RR_0^2}{R_0^2-R^2}\left(\frac{1}{r^2}-\frac{1}{R_0^2}\right)X}_\infty = 
%\norm{a\frac{R}{R_0^2-R^2}\left(\frac{R_0^2}{r^2}-1\right)X}_\infty \\ &= 
\norm{a\frac{R}{r^2}\underbrace{\left(\frac{R_0^2-r^2}{R_0^2-R^2}\right)}_{<1}X}_\infty <
\norm{a\frac{R}{r}\frac{X}{r}}_\infty = a\underbrace{\frac{R}{r}}_{<1}\underbrace{\norm{\frac{X}{r}}_\infty}_{=1}<a.
\end{aligned}
\]
%==================
\off{
One can also validate that $p$ is a subgradient element  (for the case of $R_0\rightarrow\infty$), 
\[
p\in\partial J(f) \Leftrightarrow \forall u \in BV: J(u)\geq J(f) + \langle p,u-f\rangle
\]
gives us,
\[
\begin{aligned}
\langle p,u-f\rangle &= \langle p,u\rangle - \langle p,f\rangle =  \langle \Div_A\xi^A,u\rangle - J_{A^2TV}(f) = \sup_{\xi^A,\norm{\xi^A}\le 1} \langle \Div_A\xi^A,u\rangle - J_{A^2TV}(f) \\&= J_{A^2TV}(u)- J_{A^2TV}(f) \Rightarrow p\in\partial J_{A^2TV}(f)
\end{aligned}
\]
using,
\[
\begin{aligned}
 \langle p,f\rangle
 = J_{A^2TV}(f) = ah 2 \pi R
\end{aligned}
\]
} % end off
%======================
The eigenvalue can be computed by, 
\[
J_{A^2TV}(f) = \langle p,f\rangle = \langle \lambda f,f\rangle \Rightarrow \lambda = \frac{J_{A^2TV}(f)}{\norm{f}^2} = \frac{ah2\pi R}{h^2\pi R^2} = a\frac{2}{hR}.
\]
% \subsection{Appendix II - Extension of the Numerical Experiments}
% \GG{Not sure this is needed, try to see for each figure why it important and are we not repeating things too much.
% Should be put in the thesis.
% Here - maybe only 2 figures are enough.
% }

% Further results from the numerical experiment section. 
% \begin{figure}[H]
% \centering
% \includegraphics[width = 0.8\textwidth]{Images/Non-Convex_3_1_AllCShapes_AATV.png}
% \caption{Non-Convex AllCShapes A$^2$TV.}
% \label{fig:Non-Convex_3_1_AllCShapes_AATV}
% \end{figure} 

% \begin{figure}[H]
% \centering
% \includegraphics[width = 0.8\textwidth]{Images/Non-Convex_3_2_AllCShapes_TV.png}
% \caption{Non-Convex AllCShapes TV.}
% \label{fig:Non-Convex_3_2_AllCShapes_TV}
% \end{figure} 

% \begin{figure}[H]
% \centering
% \includegraphics[width = 0.8\textwidth]{Images/Extreme-Curvature_3_1_All_Ellipse_Shapes_AATV.png}
% \caption{Extreme-Curvature All Ellipse Shapes A$^2$TV.}
% \label{fig:Extreme-Curvature_3_1_All_Ellipse_Shapes_A$^2$TV}
% \end{figure} 

% \begin{figure}[H]\centering
% \includegraphics[width = 0.8\textwidth]{Images/Extreme-Curvature_3_2_All_Ellipse_Shapes_TV.png}
% \caption{Extreme-Curvature All Ellipse Shapes TV.}
% \label{fig:Extreme-Curvature_3_2_All_Ellipse_Shapes_TV}
% \end{figure} 

\bibliographystyle{siam}
\bibliography{References}

\end{document}